%% file: GenDelProp-main.tex
\documentclass[sigconf,10pt]{acmart}

\renewcommand\footnotetextcopyrightpermission[1]{} 
\usepackage{graphicx}
\usepackage{subfigure}

\usepackage{aliascnt}
\usepackage{balance}  
\usepackage{amssymb, amsthm, amsmath, mathtools}
\usepackage{mathtools}
\usepackage[normalem]{ulem}
\usepackage{booktabs}
\usepackage{listings}
\usepackage{fancyvrb}
\usepackage{caption}
\usepackage{braket}
\usepackage{bbold}
\usepackage{bm}
\usepackage{clrscode3e}
\usepackage{colortbl}
\usepackage{forest}
\usepackage{algorithm}
\usepackage{xspace}
\usepackage{url}
\usepackage{textcomp}
\usepackage{multirow}
\usepackage{booktabs}
\usepackage{hyperref}
\usepackage{color}
\usepackage{tcolorbox}

\usepackage{tikz}
\usetikzlibrary{
    graphs,
    graphs.standard,
    calc, 
    chains,
    fit,
    positioning,
    shapes
}

\usepackage{xspace}
\usepackage{caption}
\usepackage{stackengine}
\usepackage{paralist}
\usepackage{todonotes}
\def\BibTeX{{\rm B\kern-.05em{\sc i\kern-.025em b}\kern-.08emT\kern-.1667em\lower.7ex\hbox{E}\kern-.125emX}}

\renewcommand{\footnotesize}{\small}
\graphicspath{ {imgs/} }

\input{macros.tex}

\copyrightyear{2020} 
\acmYear{2020} 
\setcopyright{acmcopyright}
\acmConference[Woodstock '18]{Woodstock '18: ACM Symposium on Neural
  Gaze Detection}{June 03--05, 2018}{Woodstock, NY}
\acmBooktitle{Woodstock '18: ACM Symposium on Neural Gaze Detection,
  June 03--05, 2018, Woodstock, NY}
\acmPrice{15.00}
\acmISBN{978-1-4503-9999-9/18/06}

\cut{
\copyrightyear{2019} 
\acmYear{2019} 
\setcopyright{acmcopyright}
\acmConference[SIGMOD '19]{2019 International Conference on Management of Data}{June 30-July 5, 2019}{Amsterdam, Netherlands}
\acmBooktitle{2019 International Conference on Management of Data (SIGMOD '19), June 30-July 5, 2019, Amsterdam, Netherlands}
\acmPrice{15.00}
\acmDOI{10.1145/3299869.3319866}
\acmISBN{978-1-4503-5643-5/19/06}
}

\begin{document}
\definecolor{myblue}{RGB}{80,80,160}
\definecolor{mygreen}{RGB}{80,160,80}

\title{Generalized Deletion Propagation on Counting Conjunctive Query Answers}

\author{Debmalya Panigrahi}
    \affiliation{%
		\institution{Duke University}
    }
\email{debmalya@cs.duke.edu}

\author{Shweta Patwa}
    \affiliation{%
		\institution{Duke University}
    }
\email{sjpatwa@cs.duke.edu}

\author{Sudeepa Roy}
    \affiliation{%
		\institution{Duke University}
    }
\email{sudeepa@cs.duke.edu}

\renewcommand{\shortauthors}{D. Panigrahi, S. Patwa, and S. Roy}

\cut{


%
}
\input{abstract.tex}




\maketitle

\begin{sloppypar}

\input{introduction.tex}

\input{prelim.tex}

\input{dichotomy.tex}

\input{approx.tex}
\input{conclusions.tex}

\balance


\bibliographystyle{ACM-Reference-Format}
\bibliography{GelDelProp-main.bbl}


\begin{appendix}

\input{appendix}


\end{appendix}
\end{sloppypar}

\end{document}

%% file: macros.tex
\newcommand{\allattr}{{\mathbb A}}
\newcommand{\aval}{{\mathbb a}}
\newcommand{\attrset}{{\mathbb B}}

\newcommand{\allrel}{{\mathbb R}}

\newcommand{\attr}{{\texttt{attr}}}
\newcommand{\rel}{{\texttt{rels}}}
\newcommand{\dom}{{\texttt{dom}}}
\newcommand{\ourprob}{{\texttt{GDP}}}
\newcommand{\true}{{\texttt{true}}}
\newcommand{\false}{{\texttt{false}}}
\newcommand{\isptime}{{\textsc{IsPtime}}}
\newcommand{\caseempty}{{\texttt{EmptyQuery-1}}} 
\newcommand{\casesingle}{{\texttt{\bf SingleRelation-2}}}
\newcommand{\casetwo}{{\texttt{\bf TwoRelations-3}}}
\newcommand{\casesubset}{{\texttt{\bf Subset-4}}}
\newcommand{\casecommon}{{\texttt{\bf CommonAttribute-5}}}
\newcommand{\casecooccur}{{\texttt{\bf CoOccurrence-6}}}
\newcommand{\casedecompose}{{\texttt{\bf Decomposition-7}}}
\newcommand{\rectree}{{{recursion-tree}}}

\newcommand{\optcost}{{\textsc{OptCost}}}
\newcommand{\optsol}{{\textsc{OptSol}}}

\newcommand{\ceil}[1]{{\lceil #1 \rceil}}

\newcommand{\computeopt}{{\textsc{ComputeOpt}}}
\newcommand{\singlerel}{{\textsc{SingleRelation}}}
\newcommand{\tworel}{{\textsc{TwoRelations}}}
\newcommand{\onesubset}{{\textsc{OneSubset}}}
\newcommand{\commonattr}{{\textsc{CommonAttrPartition}}}
\newcommand{\cooccur}{{\textsc{CoOccurrence}}}
\newcommand{\decompose}{{\textsc{DecompCrossProduct}}}

\definecolor{black}{rgb}{0,0,0}
\definecolor{grey}{rgb}{0.8,0.8,0.8}
\definecolor{red}{rgb}{1,0,0}
\definecolor{green}{rgb}{0,1,0}
\definecolor{darkgreen}{rgb}{0,0.5,0}
\definecolor{darkpurple}{rgb}{0.5,0,0.5}
\definecolor{darkdarkpurple}{rgb}{0.3,0,0.3}
\definecolor{blue}{rgb}{0,0,1}
\definecolor{shadegreen}{rgb}{0.95,1,0.95}
\definecolor{shadeblue}{rgb}{0.95,0.95,1}
\definecolor{shadered}{rgb}{1,0.85,0.85}
\definecolor{shadegrey}{rgb}{0.85,0.85,0.85}
\definecolor{oddRowGrey}{rgb}{0.80,0.80,0.80}
\definecolor{evenRowGrey}{rgb}{0.85,0.85,0.85}

\newcommand{\red}[1]{{\color{red} #1}}
\newcommand{\cut}[1]{{}}

\DeclareMathAlphabet{\mathbbold}{U}{bbold}{m}{n}

\cut{

\theoremstyle{definition}
\newtheorem{definition}{Definition}[section]

\theoremstyle{example}
\newtheorem{example}{Example}[section]
}

\newtheorem{theorem}{Theorem}[section]          	
\newaliascnt{lemma}{theorem}				
\newtheorem{lemma}[lemma]{Lemma}              	
\aliascntresetthe{lemma}  					
\newaliascnt{conjecture}{theorem}			
\aliascntresetthe{conjecture}  				
\newaliascnt{remark}{theorem}				
              
\aliascntresetthe{remark}  					
\newaliascnt{corollary}{theorem}			
\aliascntresetthe{corollary}  				
\newaliascnt{definition}{theorem}			
\newtheorem{definition}[definition]{Definition}    
\aliascntresetthe{definition}  				
\newaliascnt{proposition}{theorem}			
\aliascntresetthe{proposition}  				
\newaliascnt{example}{theorem}			
\newtheorem{example}[example]{Example}  	
\aliascntresetthe{example}  				
\newaliascnt{observation}{theorem}			
\aliascntresetthe{observation}  				

\newcommand{\eg}{{e.g.}}
\newcommand{\ie}{{i.e.}}



%% file: abstract.tex
\begin{abstract}
We investigate the computational complexity of minimizing the source side-effect in order to remove a given number of tuples from the output of a conjunctive query. In particular, given a multi-relational database $D$, a conjunctive query $Q$, and a positive integer $k$ as input, the goal is to find a minimum subset of input tuples to remove from D that would eliminate at least $k$ output tuples from $Q(D)$. This problem generalizes the well-studied deletion propagation problem in databases. In addition, it encapsulates the notion of intervention for aggregate queries used in data analysis with applications to explaining interesting observations on the output. We show a dichotomy in the complexity of this problem for the class of full conjunctive queries without self-joins by giving a characterization on the structure of $Q$ that makes the problem either polynomial-time solvable or NP-hard. Our proof of this dichotomy result already gives an exact algorithm in the easy cases; we complement this by giving an approximation algorithm for the hard cases of the problem.
\end{abstract}

%% file: introduction.tex
\section{Introduction}
\label{sec:introduction}
The problem of \emph{view update} (\eg, \cite{Bancilhon+1981, Dayal+1982}) -- how to change the input to achieve some desired changes to a \emph{view} or query output -- is a well-studied problem in the database literature. This arises in contexts where the user is interested in tuning the output to meet her prior expectation, satisfy some external constraint, or examine different possibilities. In these cases, she would want to know whether and how the input can be changed to achieve the desired effect in the output.
\par
One special case of view update that has been studied from a theoretical standpoint is \emph{deletion propagation}, which was first analyzed by Buneman, Khanna, and  Tan \cite{Buneman+2002}. Given a database $D$ and a monotone query $Q$, and a designated output tuple $t \in Q(D)$, the goal is to remove $t$ from $Q(D)$ by removing input tuples from $D$ subject to two alternative optimization criteria. In the \emph{source side-effect} version, the goal is to remove $t$ from $Q(D)$ by removing the smallest number of input tuples from $D$, whereas in the \emph{view side-effect} version, the goal is to remove $t$ such that the number of other output tuples deleted from $Q(D)$ is minimized. The intuition is that if the user considers tuple $t$ to be erroneous, then she would want to remove it from the output in a way that is minimal in terms of her intervention on the input, or in terms of the disruption caused in the output.
\par
In this paper, we study the \emph{Generalized Deletion Propagation} problem (\ourprob), where given $Q$ and $D$, instead of removing a designated tuple $t$ from $Q(D)$, the goal is to remove at least $k$ tuples from $Q(D)$ for a given integer $k$. We are interested in studying this problem from the perspective of minimizing the source side-effect, i.e., we want to remove $k$ output tuples by removing the smallest number of input tuples from $D$.  

Consider an example where an airline has flights from a set of northern locations  to a set of central locations stored in a relation ${\tt R_{nc}(north, central)}$, and also from a set of central locations to a set of southern locations stored in ${\tt R_{cs}(central, south)}$. The conjunctive query (CQ) $Q_{\rm all trips}(n, c, s):- R_{nc}(n, c), R_{cs}(c, s)$ shown in Datalog format finds all the north-central-south routes served by the airline. Now, consider a competitor that want to start new routes in a minimum number of these segments so as to affect at least $k$ of the routes being operated by the first airline. This can be exactly modeled by the $\ourprob$ problem. More generally, the $\ourprob$ problem can be used to analyze whether there is a small subset of input tuples with high impact on the output: if removal of a small number of input tuples causes a large change in the output, then there is significant dependence on this small subset which can be a potential source of vulnerability. 
\par 
The other motivation for the \ourprob\ problem comes from the recent study of \emph{explaining aggregate query answers and outliers by intervention}  \cite{WM13, RoyS14, RoyOS15}. Here, given an aggregate query $Q$, possibly with group-by operations, the user studies the outputs and may ask questions like \emph{`why a value $q_1$ is high}', or, \emph{`why a value $q_1$ is higher or lower than another value $q_2$'}. A possible explanation is a set of input tuples, compactly expressed using predicates, such that by removing these subsets we can change the selected values in the opposite direction, \eg, if the user thinks $q_1(D)$ is high, then a good explanation with high score capturing a subset $S$ of input tuples will make $q_1(D\setminus S)$ as low as possible. One example given in \cite{RoyS14} was on the DBLP publication data, where it is observed that there was a peak in SIGMOD papers coauthored by researchers in industry around year 2000 (and then it gradually declined), which is explained by some top industrial research labs that had hiring slow-down or shut down later (\ie, if the papers by these labs did not exist in the database, the peak will be lower). 
Although this line of work studies more general SQL aggregate queries with group-by and aggregates, it aims to change the output by deleting the input tuples (if the question is on a single output, it can only reduce for monotone queries). In this work, we study the complexity of the reverse direction of this problem in a simpler setting, where we only consider counts and conjunctive queries, and aim to find the minimum number of input tuples that would reduce the output by a desired amount.
\par
The \ourprob\ problem is also related to the \emph{partial vertex cover} \cite{CaskurluMPS17} or \emph{partial set cover problems} \cite{GKS04}, that are generalizations of the classical vertex cover or set cover problems, where instead of covering all edges or all elements, the goal is to select a minimum cost set of vertices or sets so that at least $k$ edges or $k$ elements are covered. These partial coverage problems are useful when the goal is to cover a certain fraction of the elements or edges, \eg, to build facilities to provide service to a certain fraction of the population \cite{GKS04}. The \ourprob\ problem is a special case of the partial set cover problem where each element is an output tuple of a CQ and each input tuple represents a set containing all the output tuples that would be removed on deleting it from the input (see Section~\ref{sec:approx}). 

\smallskip
\noindent
\textbf{Our contributions.~} We propose the \ourprob\ problem, and analyze its complexity for the class of \emph{full conjunctive queries without self-join}\footnote{The class of full CQ without self-join is a natural sub-class of CQs. Full CQs have been studied in contexts like the worst-case optimal join algorithms \cite{NgoPRR12, Veldhuizen14}, AGM bounds \cite{AGM2008}, and parallel evaluation of CQs \cite{KoutrisS11} (without self-joins). Self-join free queries have been studied in most of the related papers on deletion propagation. The complexity of \ourprob\ for larger classes of queries is interesting future work (see Section~\ref{sec:conclusions}).}. 
Given a conjunctive query (CQ) $Q$ that outputs the natural join of the input relation based on common attributes, a database instance $D$, and an integer $k$, the goal is to remove at least $k$ tuples from the output by removing the smallest number of input tuples from the database. \ourprob\ for arbitrary monotone queries $Q$ generalizes the source side-effect version of the deletion propagation problem for single or multiple output tuples, since we can add a selection operation to keep only these tuples in the output, and then run the \ourprob\ problem for $k = all$, to remove all tuples in the output. 
\par
First we give a \emph{dichotomy result} that completely resolves the complexity of the \ourprob\ problem for the class of full CQ without self-joins (Section~\ref{sec:dichotomy}). We assume the standard data complexity for our complexity results where the complexity is given in terms of the size of the input instance and the query and schema are assumed to be of constant size \cite{vardi1982complexity}.  We give an algorithm that only takes the query $Q$ as input, and decides in time that is polynomial in the size of the query (\ie, in time that is constant in data complexity), whether \ourprob\ can be solved in time that is polynomial in the data complexity for all instances $D$ and all values of $k$. If this algorithm returns true, then the problem is solvable in polynomial time for all $k$ and $D$. Moreover, if the algorithm returns false, the problem is NP-hard for some set of instances and some value of $k$. The problem we use to prove NP-hardness is \emph{partial vertex cover in bipartite graphs} (PVCB) that intends to cover at least $k$ of the edges by minimizing the cost instead of covering all the edges in a bipartite graph. Unlike the vertex cover problem in bipartite graphs, this problem was shown to be NP-hard by Caskurlu et al. \cite{CaskurluMPS17}. An example query where the reduction can be readily applied is the query for paths of length two: $Q_{2-path}(A, B) :- R_1(A), R_2(A, B), R_3(B)$ (see Lemma~\ref{lem:chain-3-nphard}). Note that this \emph{path query} was shown to be poly-time solvable for the deletion propagation problem \cite{Buneman+2002}, not only for the full CQ $Q_{2-path}$ that belongs to the class $SJ$ and therefore is poly-time solvable for a designated tuple, but also if projections are involved by a reduction to the minimum $s-t$-cut problem (the class PJ is, in general, hard for deletion propagation).  However, for arbitrary $k$, this problem becomes NP-hard for \ourprob.
\par
The query $Q$ can have more complex patterns like (attributes in the head are not displayed) 
{\footnotesize
\begin{eqnarray*}
Q_1(\cdots) & {:-} &  R_1(A), R_2(B), R_3(A, C), R_4(E, B), R_5(C, E), R_6(C, F)\\
Q_2(\cdots) & :- & R_1(A, P1, P2, E, F), R_2(B, P1, P2, E, F), R_3(P1, C1, C2),\\
& & R_4(P2, C1, C3, F), R_5(E, F, C1)
\end{eqnarray*}
}
or, a complex combination of the above two possibly involving additional attributes (we discuss these examples in Section~\ref{sec:dichotomy}). We give a set of simplification steps such that if none of them can be applied to $Q$, there is a reduction from the PVCB problem even if there is no obvious path structure like $Q_{2-path}$. In addition, we argue that the hardness is preserved in all simplification steps. 
If the algorithm to check whether a query is poly-time solvable returns true, then we give an algorithm that returns an optimal solution in polynomial time using the same simplification steps. There can be scenarios when $Q$ can be decomposed into two or more connected components, or when there is a common attribute in all relations in $Q$, and the algorithm gives a solution for each such case by building upon smaller sub-problems. 
\par
Since the \ourprob\ problem is NP-hard even for simple queries like $Q_{2-path}$, we then study approximations to this problem (Section~\ref{sec:approx}). We give an approximation algorithm by a reduction to the \emph{partial set cover} problem. When $f$ is the maximum frequency of an element in the sets, Gandhi, Khuller, and Srinivasan \cite{GKS04} generalize the classic primal dual algorithm for the set cover problem to obtain an $f$-approximation for the partial set cover problem. Using this algorithm, we get a $p$-approximation for the $\ourprob$ problem, where $p$ is the number of relations in the schema.

\smallskip
\noindent
\textbf{Related Work.~}
The classical view update problem has been studied extensively over the last four decades (\eg, \cite{Bancilhon+1981, Dayal+1982}), although the special case of deletion propagation has gained more popularity in the last two decades starting with the seminal work by Buneman, Khanna, and Tan \cite{Buneman+2002}. They showed that the class of monotone queries involving select-project-join-union (SPJU) operators can be divided into sub-classes for which finding the optimal source side-effect is NP-hard (\eg, queries with PJ or JU) or solvable in polynomial time (\eg, SPU or SJ). Recently, Friere et al. \cite{FreireGIM15} studied the \emph{resilience} problem, 
for the class of CQs without self-joins and with arbitrary functional dependency, and gave a dichotomy characterizing whether it is poly-time solvable or NP-hard. The input to the resilience problem is a Boolean CQ and a database instance $D$ such that $Q(D)$ is true, and the goal is to remove a minimum subset of tuples from the input that makes the query $Q$ evaluate to false. This is identical to the deletion propagation problem where all attributes are projected out. 
\par
In recent years,
the complexity of deletion propagation for the view side-effect version has been extensively studied by Kimelfeld, Vondrak, and Williams in a series of papers \cite{KimelfeldVW11, Kimelfeld12, KimelfeldVW13}. First, a dichotomy result was shown for CQs without self-joins \cite{KimelfeldVW11}, that if a `head-domination property' holds, then the problem is poly-time solvable; otherwise, it is APX-hard. In addition, it was shown that self-joins affect the hardness further. Then a dichotomy result was shown by Kimelfeld \cite{Kimelfeld12} for the deletion propagation problem with functional dependency for CQs without self-join. The multi-tuple deletion propagation problem was studied in \cite{KimelfeldVW13} where the goal is to remove a given set of output tuples, and a trichotomy result was shown (a query is poly-time solvable, APX-hard but constant approximation exists, or no non-trivial approximation exists). All these papers focus on the view side-effect version of deletion propagation and therefore the optimization goal is different from ours.
\par
For the source side-effect version, the complexity of multi-tuple deletion propagation was studied by Cong, Fan, and Geerts \cite{Cong+2006}. They show that for single tuple deletion propagation, \emph{key preservation} makes the problem tractable for SPJ views; however, if multiple tuples are to be deleted, the problem becomes intractable for SJ, PJ, and SPJ views. 
In our work, we study deletion propagation where the count of tuples to be removed is specified, and give a complete characterization for the class of full CQs without self joins.
\par
Beyond the context of deletion propagation, several dichotomy results have been obtained for problems motivated by data management, \eg, in the context of probabilistic databases \cite{DalviS12}, computing responsibility \cite{MeliouGMS11}, or database repair \cite{LivshitsKR18}. Problems similar to \ourprob\ have also been studied as \emph{reverse data management} \cite{MeliouGS11} where some action needs to be performed on the input data to achieve desired changes in the output. Toward this goal, Meliou and Suciu \cite{MeliouS12} studied \emph{how-to} queries, where a suite of desired changes (\eg, modifying aggregate values, creating or removing tuples) can be specified by a Datalog-like language, and a possible world satisfying all constraints and optimizing on some criteria is returned. Although \cite{MeliouS12} considered a much more general class of queries and update operations, their focus was to develop an end-to-end system using provenance and mixed integer programming, and not on the complexity of this problem. As discussed before, \ourprob\ is also related to explanations by intervention \cite{WM13, RoyS14, RoyOS15} where the goal is to find a set of input tuples captured by a predicate that changes an aggregate answer (or a function of multiple aggregate answers). For the class of simple predicates, this problem is poly-time solvable in data complexity, but complexity of the problem for more complex scenarios remains an open question.

\smallskip
\noindent
\textbf{Roadmap.~} We define some preliminary concepts in Section~\ref{sec:prelim}, then give our main dichotomy result in Section~\ref{sec:dichotomy} and the approximation results in Section~\ref{sec:approx}, and conclude with directions of future work in Section~\ref{sec:conclusions}.

%% file: prelim.tex
\section{Preliminaries}
\label{sec:prelim}

\textbf{Schema, instance, relations, attributes, tuples.~} We consider the standard setting of multi-relational databases and conjunctive queries. Let $\allrel$ be a database schema that contains $p$ tables $R_1, \cdots, R_p$.
Let $\allattr$ be the set of all attributes in the database $\allrel$. Each relation $R_i$ is defined on a subset of attributes $\attr(R_i) = \allattr_i \subseteq \allattr$. We use $A, B, C, A_1, A_2 , \cdots \in \allattr$
to denote the attributes in $\allattr$ and $a, b, c, \cdots$ etc. to denote their values. For each attribute $A \in \allattr$, $\dom(A)$ denotes the domain of $A$ and $\rel(A)$ denotes the set of relations 
that $A$ belongs to, \ie, $\rel(A) = \{R_i: A \in \allattr_i\}$. 
\par
Given the database schema $\allrel$, let $D = D^\allrel$ be a given instance of $\allrel$, and the corresponding instances of $R_1, \cdots, R_p$ be $D^{R_1}, \cdots, D^{R_p}$. Where it is clear from the context, we will use 
$D$ instead of $D^\allrel$, and $R_1, \cdots, R_p$ instead of $D^{R_1}, \cdots, D^{R_p}$. Any tuple $t \in R_i$ is defined on $\allattr_i$. For any attribute $A \in \allattr_i$, $t.A \in \dom(A)$ denotes the value of $A$ in $t$.
	Similarly, for a set of attributes $\attrset \subseteq \allattr_i$, $t.\attrset$ denotes the values of attributes in $\attrset$ for $t$ with an implicit ordering of the attributes. 
Let $n_i$ be the number of tuples in $R_i$ and $n = \sum_{i = 1}^p n_i$ be the total number of tuples in $D$. 
\begin{figure}
\begin{tabular}{|c|c|}
\multicolumn{2}{c}{$R_1$}\\
\hline
A & B \\\hline\hline
a1 & b1 \\
a2 & b1 \\
a2 & b2 \\
a3 & b3 \\\hline
\end{tabular}~~~~~~~~
\begin{tabular}{|c|c|}
\multicolumn{2}{c}{$R_2$}\\
\hline
B & C \\\hline\hline
b1 & c1 \\
b2 & c2 \\
b3 & c2 \\\hline
\end{tabular}~~~~
\begin{tabular}{|c|c|}
\multicolumn{2}{c}{$R_3$}\\
\hline
C & E \\\hline\hline
c1 & e1 \\
c1 & e2 \\
c2 & e3 \\
c3 & c3 \\\hline
\end{tabular}~~~~
\begin{tabular}{|c|c|c|c|}
\multicolumn{4}{c}{$Q(D)$}\\
\hline
A & B & C & E\\\hline\hline
a1 & b1 & c1  & e1\\
a1 & b1 & c1  & e2\\
a2 & b1 & c1  & e1\\
a2 & b1 & c1  & e2\\
a2 & b2 & c2  & e3\\
a3 & b3 & c2  & e3\\\hline
\end{tabular}
\caption{Database schema and instance from Example~\ref{eg:setup} and the answers for query $Q(A, B, C, E) :- R_1(A, B), R_2(B, C), R_3(C, E)$.}
\label{fig:example_setup}
\end{figure}

\textbf{Full conjunctive queries without self-joins.~} We consider the class of \emph{full} conjunctive queries (CQ) \emph{without self-joins}. 
Such a CQ represents the natural join among the given relations, and has the following form: 
$$Q(\allattr) :- R_1(\allattr_1), R_2(\allattr_2), \cdots, R_p(\allattr_p)$$
We will call the above query $Q$ the \emph{full CQ on schema $\allrel$}.
Note that we do not have any projection in the body or in the head of the query, and each $R_i$ in $Q$ is distinct, \ie, the CQ does not have a self-join. 
\par
When this query is evaluated on an instance $D$, the result $Q(D)$
contains all tuples $t$ defined on $\allattr$
such that there are tuples $t_i \in R_i$ with $t_i.A = t.A$ for all attributes $A \in \allattr_i$, for all $i = 1, \cdots, p$.
Extending the notations, we use $\rel(Q)$ to denote all the relations that appear in the body of $Q$ (initially, $\rel(Q) = \allrel$), and $\attr(Q)$ to denote all the attributes 
that appear in the body of $Q$ (initially, $\attr(Q) = \allattr$).

\begin{example}\label{eg:setup}
In Figure~\ref{fig:example_setup}, we show an example database schema $\allrel$ with three relations $R_1, R_2, R_3$, where $\allattr = \{A, B, C, E\}$, $\allattr_1 = \{A, B\}$, $\allattr_2 = \{B, C\}$, and $\allattr_3 = \{C, E\}$. 
Further, $\rel(A) = \{R_1\}$, $\rel(B) = \{R_1, R_2\}$, $\rel(C) = \{R_2, R_3\}$, and $\rel(E) = \{R_3\}$. It also shows an instance $D$ and the result $Q(D)$ of the CQ $Q(A, B, C, E) :- R_1(A, B), R_2(B, C), R_3(C, E)$. Here $n = 11$ and $p = 4$. 
\end{example}

\textbf{Generalized Deletion Propagation problem \ourprob).~} Below we define the generalized deletion propagation problem in terms of the count of output tuples of a CQ:

\begin{definition}\label{def:problem}
Given a database schema $\allrel$ with $p$ relations $R_1, \cdots, R_p$, a CQ $Q$ on $\allrel$, an instance $D$, and a positive integer $k \geq 1$, the \emph{generalized deletion propagation problem (\ourprob)}  aims to remove at least $k$ tuples from the output $Q(D)$ by removing the minimum number of input tuples from $D$.
\end{definition}
Given $Q$, $k$, and $D$, we denote the above problem by $\ourprob(Q, k, D)$ (note that the schema is implicit in $Q$).
\begin{example}\label{eg:GDP}
Suppose $k = 4$ for the input in Example~\ref{eg:setup}. Then, given the instance in Figure~\ref{fig:example_setup}, the solution of \ourprob\ will include a single tuple $R_2(b1, c1)$ since by removing this tuple we would remove the first four output tuples in $Q(D)$.  
\end{example}
For arbitrary CQs, \ourprob\ generalizes the deletion propagation problem (both single- and multi-tuple versions), since we can only select the intended tuple(s) for deletion by a selection operation at the end, and then run \ourprob\ for $k = all$.\\

\textbf{\ourprob\ for full CQ without self-joins.~} In this paper we study the complexity of \ourprob\ for the class of full CQ without self-joins. 
Note that the problem is trivial if $k = 1$: since we do not allow projection, any output tuple can be removed by removing any one input tuple that has been used to produce the output tuple. 
This is observed in \cite{Buneman+2002}, who identified the class of SJ queries as poly-time solvable for single-tuple deletion propagation.\\ 

\textbf{Data complexity.~} In this paper we assume standard data complexity \cite{vardi1982complexity}, where the size of the input instance $n$ is considered variable, but the size of the query and schema is assumed to be constant, \ie, $p, |\allattr|$ are constants.


%% file: dichotomy.tex
\section{Dichotomy}
\label{sec:dichotomy}

In this section, we give the following dichotomy result that characterizes the complexity of \ourprob\ on the full CQ of any input schema $\allrel$ with $p$ relations $R_1, \cdots, R_p$.
\begin{theorem}\label{thm:dichotomy}
If the algorithm $\isptime(Q)$ given in 
Algorithm~\ref{algo:dichotomy} returns \true, then for \emph{all} values of integer $k$ and for \emph{all} instances $D$, the problem $\ourprob(Q, k, D)$ is poly-time solvable in data complexity. 
Moreover, an optimal solution can be computed in poly-time. 
Otherwise the problem   $\ourprob(Q, k, D)$ is NP-hard.  
\end{theorem}
\begin{algorithm}[t]\caption{Deciding whether \ourprob\ for query $Q$ is poly-time solvable for all $k$}\label{algo:dichotomy}
{\footnotesize
	\begin{codebox}
		\Procname{\isptime(Q)}
		\li \If $\rel(Q) = \emptyset$ or $\attr(Q) = \emptyset$  \hfill \textbf{/* (\caseempty) */}
		\li \Do \Return \true
		\li \ElseIf $Q$ has one relation  \textbf{/* (\casesingle) */}
		\li \Do \Return \true
		\li \ElseIf $Q$ has two relations  \textbf{/* (\casetwo) */}
		\li \Do \Return \true
		\li \ElseIf $\exists R_i \in \rel(Q)$ such that $\forall~ R_j \neq R_i \in \rel(Q)$, \\
		$\attr(R_i) \subseteq \attr(R_j)$  \textbf{/* (\casesubset) */}
		\li \Do \Return \true
		\li \ElseIf $\exists A \in \attr(Q)$ such that \\ for all relations $R_i \in \rel(Q)$, $A \in \attr(R_i)$ \\  \textbf{/* (\casecommon) */}
		\li \Do Let $Q_{-A}$ be the query formed by removing $A$ \\from each relation in $\rel(Q)$
		\li \Return $\isptime(Q_{-A})$   
		\li \ElseIf  $\exists A, B \in \attr(Q)$ such that $\rel(A) = \rel(B)$ \\  \textbf{/* (\casecooccur) */}
		\li \Do Replace both $A, B$ by a new attribute $C \notin \attr(Q)$\\
		 in all relations where $A$ and $B$ appear
		\li Let the new query be $Q_{AB \rightarrow C}$
		\li \Return \isptime($Q_{AB \rightarrow C}$)   
		\li \ElseIf $Q$ can be decomposed into \emph{maximal connected components}\\ (see text) $Q^1, \cdots, Q^s$ where $s \geq 2$ \textbf{/* (\casedecompose) */}
		\li \Do \Return $\wedge_{i = 1}^s \isptime(Q^i)$
		\li \Else \Return \false	
	\end{codebox}
	}
\end{algorithm}
The algorithm has seven simplification steps as written next to each condition, and some simplification steps call the algorithm $\isptime$ recursively\footnote{It may be noted that Algorithm~\ref{algo:dichotomy} has a correspondence with \emph{hierarchical queries}, where for any two attributes $A, B$, either one of $\rel(A), \rel(B)$ is a subset of the other, or they are disjoint. Hierarchical queries have been used in the  seminal dichotomy results of efficient query evaluations in probabilistic databases by Dalvi and Suciu (\eg, \cite{DalviS12}) that classify queries either as poly-time solvable or \#P-hard. Our problem is an optimization problem instead of a counting-like problem of query evaluation in probabilistic databases, and the proof techniques for both hardness and algorithmic results are different.}. The first four steps check if the query is empty, has one or two relations, or there is a relation whose all attributes appear in all other relations; then the algorithm returns true. The fifth step looks for a common attribute present in all relations, and the sixth step checks whether two attributes co-occur in all relations. The last simplification step 
decomposes the query $Q$ into two or more \emph{maximal connected components} (if possible), which can be achieved by a standard procedure. We form a graph $G_Q$ on $\rel(Q)$ as the vertices. For any two relations $R_i, R_j \in \rel(Q)$, if there is an attribute $A \in \attr(R_i) \cap \attr(R_j)$,  then we add an edge between $R_i$ and $R_j$ in $G_Q$. Then we decompose $G_Q$ into maximal connected components using standard graph-traversal-based algorithms \cite{Cormen:2009} and call the components as $Q^1, \cdots, Q^s$. For instance, if $Q(A, B, C, E, F, G) :- R_1(A, B), R_2(F), R_3(B, C), R_4(G), R_5(C, E)$, then $Q$ can be decomposed into $s = 3$ maximal connected components $Q^1(A, B, C) :- R_1(A, B), R_3(B, C), R_5(C, E)$, $Q^2(F) :- R_2(F)$, and $Q^3(G) :- R_4(G)$.
\par
Before we prove Theorem~\ref{thm:dichotomy}, we give some examples illustrating the application of the theorem (we omit the attributes in the head of the queries).
\begin{example}\label{eg:queries}
\begin{itemize}
    \item Consider $Q_0(\cdots):-$ $R_1(A, B)$, $R_2(F, G)$, $R_3(B, C, E)$, $R_4(C, E)$, $R_5(G, H)$ (also see Figure~\ref{fig:example_T}). Observe that the first four simplifications cannot be applied to $Q_0$. Since $rels(C)=rels(E)$, \casecooccur\ gives $Q_{0 \hspace{1mm} C,E\rightarrow K}$. Next, \casedecompose\ is applied which gives $Q^1$ (with $R_1, R_3, R_4$) and $Q^2$ (with $R_2, R_5$). The query $Q^2$ has two relations so it returns \true. However, All simplifications fail for $Q^1$, so it returns {\false}. In turn, \isptime\ for $Q_{0 \hspace{1mm} C,E\rightarrow K}$ and $Q_0$ return \false. Therefore $Q_0$ is a hard query. 
    \item Consider the queries $Q_1$ and $Q_2$ given in the introduction. None of the simplification steps can be applied to both these queries, so these two queries are NP-hard (albeit the NP-hardness for these two queries are shown by two different proof techniques as shown Lemma~\ref{lem:nphard-disjoint} and Lemma~\ref{lem:nphard-overlap}). 
    \item Consider $Q_3(\cdots)$ $:-$ $R_1(A, P)$, $R_2(A, P, M)$, $R_3(A, P, G)$, $R_4(A, B, C)$, $R_5(A, C, E)$, $R_6(F)$. Here first we apply \casedecompose\ to separate $R_6$, which is a single relation, and therefore returns \true. For $R_1, \cdots, R_5$, first \casecommon\ is applied to remove $A$. Then \casedecompose\ is again applied to partition into $R_1, R_2, R_3$ and $R_4, R_5$. The intermediate query for $R_4, R_5$ returns \true\ as there are two relations. For $R_1, R_2, R_3$, \casesubset\ is applied and returns true as $P$ is the only remaining attribute in $R_1$ and it belongs both $R_2$ and $R_3$. In turn, all the intermediate queries and the original query return \true. By Theorem~\ref{thm:dichotomy}, this query is poly-time solvable, and an optimal solution can be obtained by applying Algorithm~\ref{algo:computeopt} in Section~\ref{sec:algo}. 
\end{itemize}
\end{example}

In our hardness proofs and in our algorithms, we use the \emph{\rectree} of query $Q_0$ capturing the repeated use of $\isptime(Q)$ for intermediate queries $Q$ within $\isptime(Q_0)$, which is defined as follows.

\begin{definition}\label{def:rectree}
Let $Q_0$ be the query that is given as the initial input to Algorithm~\ref{algo:dichotomy}. The \rectree\ of $Q_0$ is a tree $T$ where the non-leaf nodes denote intermediate CQs $Q$ for which $\isptime(Q)$ has been invoked during the execution of $\isptime(Q_0)$. The leaves in $T$ are \true\ or \false. The root is $Q_0$. 
\par
If  one of the first four steps is applied to an intermediate query $Q$, we add a single leaf child to $Q$ and assign \true. If no simplifications can be applied, we add a leaf node with value $\false$ as the child of query $Q$. If \casecooccur\ is applied in Algorithm~\ref{algo:dichotomy}, we add  a single child $Q_{A, B \rightarrow C}$ of $Q$ (see Algorithm~\ref{algo:dichotomy}). If \casecommon\ is applied, we add a single child $Q_{-A}$ of $Q$. If \casedecompose\ is applied, we add $s$ children $Q^1, \cdots, Q^s$ where $s$ is the number of connected components. 
\end{definition}
{\footnotesize
\begin{figure}[!ht]
    \centering
    \begin{tikzpicture}[sibling distance=12em, level distance=5em
     ]]
        \node[text width=65mm] {$Q_0(A, B, C, E, F, G, H):-$ $R_1(A, B)$, $R_2(F, G)$, $R_3(B, C, E)$, $R_4(C, E)$, $R_5(G, H)$}
        child{ 
            node[text width=65mm] {$Q_{0 \hspace{1mm} C, E\rightarrow K}(A, B, F, G, H, K):-$ $R_1(A, B)$, $R_2(F, G)$, $R_3(B, K)$, $R_4(K)$, $R_5(G, H)$} 
            child{ 
                node[text width=45mm] {$Q^1(A, B, K):-$ $R_1(A, B)$, $R_3(B, K)$, $R_4(K)$}
                child{
                    node{{\em false}}
                }edge from parent node[left] {7}
            }
            child {
                node[text width=30mm] {$Q^2(F, G, H):-$ $R_2(F, G)$, $R_5(G, H)$}
                child{
                    node{{\em true}}edge from parent node[right] {2}
                }edge from parent node[right] {7}
            }edge from parent node[left] {6}
        }; 
    \end{tikzpicture}

    \caption{Tree $T$ when Algorithm \ref{algo:dichotomy} is run on $Q_0$ from Example~\ref{eg:queries}, which simplification step has been applied is shown next to each edge.} 
    \label{fig:example_T}
\end{figure}
}

\cut{
\par

Now consider $Q(A, B, C, D):-R_1(A, B), R_2(C, D)$. The full CQ $Q$ will result in a cross product between the two relations. This schema is another example of an easy case for $\ourprob$. Here, we can compute in poly-time the number of output tuples resulting from an input tuple in either relation, and greedily pick the ones that intervene on the most number of output tuples. This gives us a smallest set of input tuples to delete to eliminate at least $k$ tuples from $Q(D)$.

In the case that the two relations, $R_1$ and $R_2$, share attributes, we can first partition the relations by the combination of values of the common attributes, and then build an overall optimal solution by combining choices from the different parts. We give details in Algorithm~\ref{algo:tworel}. The idea of partitioning relations is also important when an attribute appears in all relations (see Algorithm~\ref{algo:commonattr}). Observe that $Q(A, B, C, E):-R_1(A, B), R_2(A, C), R_3(A, E)$, also known as {\em star join}, is a special case.

When two attributes, say $A$ and $B$, appear in the same set of relations in CQ $Q$, we can combine $A$ and $B$ by replacing them with a new attribute, say $C$, so that the new tuple $t'$ is of the form $t'.C=(t.A, t.B)$, for every tuple $t$ in $rels(A)\cap rels(B)$ (see Algorithm~\ref{algo:cooccur}).

However, $Q(A, B):-R_1(A), R_2(A, B), R_3(B)$ is a hard case for $\ourprob$, which we prove in Lemma~\ref{lem:chain-3-nphard}. This is the simplest setting of a hard case. We show how this ties into more complicated schemas in Lemmas~\ref{lem:nphard-disjoint}-\ref{lem:hardness}.
}

We prove Theorem \ref{thm:dichotomy} in the next two subsections. First, in Section~\ref{sec:hardness} we show that if Algorithm~\ref{algo:dichotomy} returns false for a query $Q$, then $\ourprob(Q, k, D)$ is NP-hard for $Q$. Then in Section~\ref{sec:algo} we give algorithms to solve the $\ourprob(Q, k, D)$ problem in polynomial time for all $k$ and $D$ when $\isptime(Q)$ returns true. 

\subsection{Hardness}\label{sec:hardness}
The proof of hardness is divided into two parts. In the first part we argue that for the last three simplification steps that call the $\isptime$ algorithm recursively, hardness of the query is preserved ({\bf proved in Appendix~\ref{sec:hard-propagate}} due to space constraints). In particular, the last step \casedecompose\ needs a careful construction here to ensure a poly-time reduction, which we prove in two steps. First we show that an immediate application of \casedecompose\ gives a poly-time reduction showing NP-hardness, but it leads to a polynomial increase in the instance size. Then we show how the NP-hardness can be obtained up to the original query $Q_0$ still incurring a polynomial increase in the problem size. In the first four simplification steps the algorithm returns true, so we do not need to consider them.
\par
In the second part we argue that if all the simplification steps fail, then the problem is NP-hard, which we discuss below. 

\cut{
\subsubsection{Hardness propagation for simplification steps}\label{sec:hard-propagate}
While the proofs for the fifth and sixth steps are more intuitive, the proof for the seventh step needs a careful construction. 

\textbf{Hardness propagation \casecommon .~} The following lemma shows the hardness propagation for the fifth simplification step  in Algorithm~\ref{algo:dichotomy}.
\begin{lemma} \label{lem:hard-common}
Let $\exists A \in \attr(Q)$ such that for all relations $R_i \in \rel(Q)$, $A \in \attr(R_i)$, and 
let $Q' = Q_{-A}$ be the query formed by removing $A$ from each relation in $\rel(Q)$. If $\ourprob(Q', k, D')$ is NP-hard, 
then $\ourprob(Q, k, D)$ is NP-hard. 
\end{lemma}
\begin{proof}
Given an instance of $\ourprob(Q', k, D')$, we construct an instance of $\ourprob(Q, k, D)$ as follows. 
\par
First, we claim that no relation $R_i$ in $Q'$ can have empty set of attributes. Otherwise, $A$ was the single attribute in $R_i$, which also appeared in all other relations in $Q$, \ie, $\forall R_j,~ \attr(R_i) \subseteq \attr(R_j)$. Therefore, the condition for \casesubset\ is satisfied, which is checked before \casecommon, and would have returned \true. Therefore, all relations in $Q'$ has at least one attribute.
\par
For any relation $R_i' \in \rel(Q')$, if a tuple $t'$ appears in $D'^{R_i'}$, create a new tuple $t$ in $D^{R_i}$ such that $t.A = *$ (a fixed value for all tuples and all relations in attribute $A$), and for all other attributes $E$, $t.E = t'.E$ (there is at least one such $E$). 
Hence there is a one-to-one correspondence between the tuples in the output $Q(D)$ and $Q'(D')$, and also in the input $D$ and $D'$. Therefore, a solution to $\ourprob(Q, k, D)$ of size $C$ corresponds to a solution to $\ourprob(Q', k, D')$, and vice versa. \end{proof}


\textbf{Hardness propagation for \casecooccur.~} Next we show the hardness propagation for the sixth simplification step in Algorithm~\ref{algo:dichotomy}.
\begin{lemma}\label{lem:hard-cooccur}
Let $A, B \in \attr(Q)$ such that $\rel(A) = \rel(B)$. Let $Q' = Q_{AB \rightarrow C}$ be the query by replacing $A, B$ with a new attribute $C \notin \attr(Q)$ in all relations. If $\ourprob(Q', k, D')$ is NP-hard, 
then $\ourprob(Q, k, D)$ is NP-hard.
\end{lemma}
\begin{proof}
Given an instance of $\ourprob(Q', k, D')$, we construct an instance of $\ourprob(Q, k, D)$ as follows. Consider any relation $R_i' \in \rel(Q')$ such that $C \in \attr(R_i')$, the corresponding relation $R_i$ in $Q$ has both attributes $A, B$ instead of $C$. If a tuple $t'$ appears in $D'^{R_i'}$, create a new tuple $t$ in $D^{R_i}$ such that $t.A = t.B = t'.C$, and for all other attributes $E$, $t.E = t'.E$ (\ie, both $A, B$ attributes get the value of attribute $C$ in $D$). Hence there is a one-to-one correspondence between the tuples in the output $Q(D)$ and $Q'(D')$, and also in the input $D$ and $D'$. Therefore, a solution to $\ourprob(Q, k, D)$ of size $C$ corresponds to a solution to $\ourprob(Q', k, D')$, and vice versa. 
\end{proof}

\textbf{Hardness propagation for \casedecompose.~} Now we show the hardness propagation for the seventh simplification step. Unlike the above two steps, this step requires a careful construction. For instance, consider a query  $Q(A, B, E):- R_1(A), R_2(A, B), R_3(B), R_4(E)$, which can be decomposed into two connected components $Q^1(A, B) :- R_1(A), R_2(A, B), R_3(B)$ and $Q^2(E) :- R_4(E)$. As we show later in Lemma~\ref{lem:chain-3-nphard}, $\ourprob(Q^1, k', D')$ is NP-hard for some $k', D'$, so although $Q^2$ is easy (see Section~\ref{sec:algo}), $\ourprob(Q, k, D)$ should be NP-hard for some $k, D$. An obvious approach is to assign a dummy value for $E$ in $R_4(E)$ similar to Lemma~\ref{lem:hard-common} above. However, if the number of tuples in $R_4$ is one or small in $D$, $\ourprob(Q, k, D)$ gains advantage by removing tuples from $R_4$, thereby completely bypassing $Q^1$. Therefore, a possible solution is to use a large number of tuples in $R_4$ that do not give a high benefit to delete from $R_4$, \eg, if it has more than the number of tuples in the output of $Q^1$ on $D$ restricted to $R_1, R_2, R_3$. First, we show the hardness propagation for a single application of \casedecompose. 
 
\begin{lemma} \label{lem:hard-decompose}
Let $Q$ is decomposed into maximal connected components  $Q^1, \cdots, Q^s$ where $s \geq 2$. without loss of generality (wlog.), suppose $\ourprob(Q^1, k', D')$ is NP-hard. 
Then, $\ourprob(Q, k, D)$ is NP-hard.
\end{lemma}
\begin{proof}
Given $Q^1, k', D'$, we create $k, D$ as follows. In $D$, all relations in $Q^1$ retain the same tuples. For all relations $R_i \in \bigcup_{j = 2}^s \rel(Q^j)$, we create $L$ tuples as follows: let us fix $R_i$, and let $A_1, \cdots, A_u$ be the attributes in $R_i$. $R_i$ in $D$ contains $L$ tuples of the form $t_{\ell} = (a_{1, \ell}, a_{2, \ell}, \cdots, a_{u, \ell})$, for $\ell = 1$ to $L$. Similarly, we populate the other relations. Note that in all the relations $R_i$ that an attribute $A_h$ appears in, it has $L$ values $a_{h, 1}, \cdots, a_{h_L}$. Therefore, any connected component $Q^2, \cdots, Q^s$ except $Q^1$ has $L$ output tuples (the components are maximally connected) and each input tuple of a relation participates in exactly one output tuple \emph{within} the connected component. Since $Q^1, \cdots, Q^s$ are disjoint in terms of attributes, in the output of $Q$, the outputs of each connected component will join in cross products. Suppose $Q^1(D')$ has $P$ output tuples. Then the number of output tuples in $Q(D)$ is $P \cdot L^{s-1}$. We set $k = k' \cdot L^{s-1}$ and $L = P+1$. The size of $D$ is $|D'| + L(s-1)$. Since $P \leq |D'|^{p'}$ (where $p'$ is the number of relations in $Q^1$), the increase in size of the inputs in this reduction is still polynomial in data complexity. 
Now we argue that $\ourprob(Q^1, k', D')$ has a solution of size $C$ if and only if $\ourprob(Q, k, D)$ has a solution of size $C$.
\par
(only if) If by removing $C$ tuples from $\rel(Q^1)$ we remove $k'$ tuples from $Q^1(D')$, then by removing the same  $C$ tuples we will remove $k' \cdot L^{s-1}$ output tuples from $Q(D)$ by construction as the output tuples from the connected components join by cross product, and each connected component has $L$ output tuples. 
\par
(if) Consider a solution to $\ourprob(Q, k, D)$ that removes at least $k = k'. L^{s-1}$ tuples from $Q(D)$. Note that any tuple from any relation $R_i \in \rel(Q^j)$, $j \geq 2$, can remove exactly 1 output tuple from the output of connected component $Q^j$ that it belongs to. Therefore, it removes exactly $o_2 = P. L^{s-2}$ tuples from the output. On the other hand, since we do not have any projection, any tuple from any relation $R_i \in \rel(Q^1)$ removes at least one output tuple from $Q^1(D')$, therefore at least $o_1 = L^{s-1}$ output tuples from $Q(D)$. Since $L = P+1$, $o_2 < o_1$. Therefore, if the assumed solution to $\ourprob(Q, k, D)$ removes any input tuple from any relation belonging to $Q^2, \cdots, Q^s$, we can replace it by any input tuple from the relations in $Q^1$ that has not been removed yet without increasing the cost or decreasing the number of output tuples removed. Therefore, wlog. all removed tuples appear in relations in $Q^1$. Since $k = k'. L^{s-1}$ tuples are removed from $Q(D)$, each tuple in $Q^1(D')$ removes exactly $L^{s-1}$ tuples from $Q(D)$, and the set of tuples removed from $Q(D)$ by tuples from $Q^1(D')$ are disjoint, at least $k'$ tuples must be removed from $Q^1(D')$ which gives a solution of cost at most $C$. 
\end{proof}

Although the above proof requires an exponential blow-up in the size of the query and not the data, there may be multiple application of \casedecompose\ in Algorithm~\ref{algo:dichotomy} in combination with the other simplification steps. Therefore, we need to ensure that the size of the instance $D$ that we create from $D'$ is still polynomial in data complexity for the original query that we started with. 
\par
Using ideas from Lemmas~\ref{lem:hard-common}, \ref{lem:hard-cooccur}, and \ref{lem:hard-decompose}, below we argue if any application of \casedecompose\ yields a hard query in one of the components, then the 	query we started with (say $Q_0$) is hard. Such an argument was not needed for   \casecommon\ and \casecooccur\ since in Lemma~\ref{lem:hard-common} and \ref{lem:hard-cooccur} the reductions do not yield an increase in the size of database instance. 
\par

\begin{lemma}\label{lem:hard-decompose-general}
Let $Q_0$ be the query that is given as the initial input to Algorithm~\ref{algo:dichotomy}. For any intermediate query $Q'$ in the \rectree\ of $Q_0$, if $\ourprob(Q', k', D')$ is NP-hard, 
then $\ourprob(Q_0, k, D)$ is NP-hard. 
\end{lemma} 
\begin{proof}
Consider the \rectree\   $T$ in which the simplification steps have been applied from $Q_0$ to $Q'$. 
Now consider the node $Q'$ in $T$ such that $\ourprob(Q', k', D')$ is NP-hard. The instance $D'$ is defined on the relations and attributes in $Q'$. From $D'$, we need to construct an instance $D$ on the relations and attributes in $Q_0$.  
\par
Consider the path from $Q'$ to the root $Q_0$. The relations in $Q'$ can lose attributes from the corresponding relations in $Q_0$ only by steps \casecommon\ and  \casecooccur\ along this path. 
For the the attributes that were lost on the path from $Q_0$ to $Q'$, we populate the values bottom-up from $Q'$ to $Q_0$ as follows. The relations appearing in $Q'$ have the same number of tuples in $D$ and $D'$. Moreover, (i) if two variables $A, B$ are replaced by a variable $C$  by \casecooccur, both $A$ and $B$ get the same values of $C$ in the corresponding tuples, (ii) if a variable $A$ is removed by \casecommon, we replace it by a constant value $*$ in all tuples. Let $Q$ be the query formed by extending the relations in $Q'$ with attributes by this process at the root.  
\par
Note that the relations in any non-descendant and non-ancestor node $Q_{nad}$ will be disjoint from those in $Q'$, but they can share some attributes with $Q$ only  by \casecommon: \casecommon\ is applied before \casecooccur\ so multiple attributes co-occurring in all relations will be removed by \casecommon\ not by \casecooccur; further before \casecooccur\ can be applied, the decomposition step \casedecompose\ must be called at least once.  We take all the relations that do not appear in the ancestors and descendants of $Q'$, and do a maximal connected component decomposition on them excluding the attributes that are common with $Q$. Let $s$ be the number of connected components. 
The tables in the connected components each get $L$ tuples as in the construction of Lemma~\ref{lem:hard-decompose}: consider a relation $R_i  \notin \rel(Q)$. (a) if there is an attribute $A \in \attr(R_i) \cap \attr(Q)$, assign $A = *$ in all $L$ tuples, (b) for all other attributes say $(A_1, \cdots, A_u) \in \attr(R_i) \setminus \attr(Q)$,  $R_i$ in $D$ contains $L$ tuples of the form $t_{\ell} = (a_{1, \ell}, a_{2, \ell}, \cdots, a_{u, \ell})$, for $\ell = 1$ to $L$. Therefore, the number of output tuples in each connected component is $L$ and each input tuple from each connected component can remove exactly one output tuple from the component. 

\begin{figure}[!ht]
    \centering
    \begin{tabular}{|c|c|}
        \multicolumn{2}{c}{$R_1$}\\
        \hline A & B\\\hline\hline
        $a_1$ & $b_1$\\
        $a_1$ & $b_2$\\
        $a_2$ & $b_1$\\
        \hline
    \end{tabular}
    \qquad
    \begin{tabular}{|c|c|}
        \multicolumn{2}{c}{$R_3$}\\
        \hline B & K\\\hline\hline
        $b_1$ & $k_1$\\
        $b_1$ & $k_2$\\
        $b_2$ & $k_1$\\
        $b_2$ & $k_3$\\
        \hline
    \end{tabular}
    \qquad
    \begin{tabular}{|c|}
        \multicolumn{1}{c}{$R_4$}\\
        \hline K\\\hline\hline
        $k_1$\\
        $k_2$\\
        $k_3$\\
        $k_4$\\
        \hline
    \end{tabular}
    \\ $Q^1(A, B, K):-$ $R_1(A, B)$, $R_3(B, K)$, $R_4(K)$\\
    
    \begin{tabular}{|c|c|}
        \multicolumn{2}{c}{$R_1$}\\
        \hline A & B\\\hline\hline
        $a_1$ & $b_1$\\
        $a_1$ & $b_2$\\
        $a_2$ & $b_1$\\
        \hline
    \end{tabular}
    \begin{tabular}{|c|c|c|}
        \multicolumn{2}{c}{$R_3$}\\
        \hline B & C & E\\\hline\hline
        $b_1$ & $k_1$ & $k_1$\\
        $b_1$ & $k_2$ & $k_2$\\
        $b_2$ & $k_1$ & $k_1$\\
        $b_2$ & $k_3$ & $k_3$\\
        \hline
    \end{tabular}
    \begin{tabular}{|c|c|}
        \multicolumn{2}{c}{$R_4$}\\
        \hline C & E\\\hline\hline
        $k_1$ & $k_1$\\
        $k_2$ & $k_2$\\
        $k_3$ & $k_3$\\
        \hline
    \end{tabular}
    \begin{tabular}{|c|c|}
        \multicolumn{2}{c}{$R_2$}\\
        \hline F & G\\\hline\hline
        $f_1$ & $g_1$\\
        $f_2$ & $g_2$\\
        $f_3$ & $g_3$\\
        $f_4$ & $g_4$\\
        $f_5$ & $g_5$\\
        $f_6$ & $g_6$\\
        $f_7$ & $g_7$\\
        \hline
    \end{tabular}
    \begin{tabular}{|c|c|}
        \multicolumn{2}{c}{$R_5$}\\
        \hline G & H\\\hline\hline
        $g_1$ & $h_1$\\
        $g_2$ & $h_2$\\
        $g_3$ & $h_3$\\
        $g_4$ & $h_4$\\
        $g_5$ & $h_5$\\
        $g_6$ & $h_6$\\
        $g_7$ & $h_7$\\
        \hline
    \end{tabular}
    \caption{An example construction for Lemma~\ref{lem:hard-decompose-general}: Consider the full CQ $Q_0$ from Figure \ref{fig:example_T}. For the instance of $Q^1$ given in this figure, we create an instance $D$ for $Q_0$. 
    Since we replaced $C,E$ with $K$,
    both $C, E$ get the values from $K$. The output size of $Q^1$ is 6, therefore the relations in $Q^2$ that do not share any attributes with $Q^1$, so we add $6 + 1 = 7$ tuples to $R_2$ and $R_5$ of the form $(f_i, g_i)$ and $(g_i, h_i)$, for $1\leq i\leq 7$, respectively.}
    \label{fig:hard-decompose-general}
\end{figure}
An example reduction is shown in Figure~\ref{fig:hard-decompose-general}.
\par
Eventually, $Q_0$ is formed by joining $Q$ with the relations in the $s$ connected components. The attribute values $*$ or repeated values due to \casecommon\ and \casecooccur\ are not going to impact the number of output tuples.
\par
Now the same reduction as in Lemma~\ref{lem:hard-decompose} works: we set $k = k'. L^s$ where $L = P+1$, and $P = $ the number of tuples in $Q'(D')$. 
We again argue that $\ourprob(Q', k', D')$ has a solution of size $C$ if and only if $\ourprob(Q_0, k, D)$ has a solution of size $C$.
\par
(only if) If by removing $C$ tuples from $\rel(Q')$ we remove $k'$ tuples from $Q'(D')$, then by removing the corresponding $C$ tuples from $Q$, we will remove $k' \cdot L^s$ output tuples from $Q_0(D)$, by construction. The output tuples from the $s$ connected components join by cross product, and each connected component has $L$ output tuples. 
\par
(if) Consider a solution to $\ourprob(Q_0, k, D)$ that removes at least $k = k'. L^{s}$ tuples from $Q_0(D)$. Note that any tuple from any relation $R_i \notin \rel(Q)$, can remove exactly 1 output tuple from the output of connected component that it belongs to. Therefore, it removes exactly $o_2 = P. L^{s-1}$ tuples from the output. On the other hand, since we do not have any projection, any tuple from any relation $R_i \in \rel(Q)$ removes at least one output tuple from $Q'(D')$, therefore at least $o_1 = L^{s}$ output tuples from $Q_0(D)$. Since $L = P+1$, $o_2 < o_1$. Therefore, if the assumed solution to $\ourprob(Q_0, k, D)$ removes any input tuple from any relation belonging to the relations $\notin \rel(Q)$, we can replace it by any input tuple from the relations in $Q$ that has not been removed yet without increasing the cost or decreasing the number of output tuples removed. Therefore, wlog. all removed tuples appear in relations in $Q$. Since $k = k'. L^{s}$ tuples are removed from $Q_0(D)$, each tuple in the relations from $Q$ removes exactly $L^{s}$ tuples from $Q_0(D)$. Since the set of tuples removed from $Q_0(D)$ by tuples from $Q$ are disjoint, and the extension of attributes from relations in $Q'$ to those in $Q$ by repeating values or by using a constant $*$ does not have an effect on the number of output tuples, at least $k'$ tuples must be removed from $Q'(D')$, giving a solution of cost at most $C$. 
\end{proof}
}
The NP-hard problem that we use to prove the NP-hardness of \ourprob\ is  the \emph{partial vertex cover problem for bipartite graphs (PVCB)} defined as follows.
\begin{definition}\label{def:pvc}
The input to the PVCB problem is an undirected bipartite graph $G(U, V, E)$ where $E$ is the set of edges between two sets of vertices $U$
 and $V$, and an integer $k$. The goal is to find a subset $S \subseteq U \cup V$ of minimum size such that at least $k$ edges from $E$ have at least one endpoint in $S$. The PVCB problem has been shown to be NP-hard (\cite{CaskurluMPS17})\footnote{The value of $k$ in the hard instance of PVCB in \cite{CaskurluMPS17} is not a constant, hence the complexity of \ourprob\ for constant $k > 1$ remains open.}.
 \end{definition}

\cut{
Recall from our earlier discussion that {\em path query} has been shown to be poly-time solvable for the deletion propagation problem \cite{Buneman+2002}. The goal there is to remove a target output tuple, and Buneman et al. model it as a minimum $s-t$-cut problem. However, the goal in $\ourprob(Q, k, D)$ is to find the minimum subset of tuples in $Q$ to remove at least $k$ output tuples from $Q(D)$. We first prove that $Q_{2-path}(A, B)$ is NP-hard (Lemma~\ref{lem:chain-3-nphard}), and then show how it ties into more complex, hard schemas for $\ourprob(Q, k, D)$ (Lemmas~\ref{lem:nphard-disjoint}-\ref{lem:hardness}).

\begin{lemma}\label{lem:chain-3-nphard}
For the query $Q_{2-path}(A, B) :- R_1(A), R_2(A, B), R_3(B)$, the problem $\ourprob(Q, k, D)$ is NP-hard.
\end{lemma}
\begin{proof}
We give a reduction from PVCB problem that takes as input $G=(U, V, E)$ and $k$.
 \par
 Given an instance of the PVCB problem, we construct an instance $D$ of \ourprob\ as follows for  $Q_{2-path}(A, B) :- R_1(A), R_2(A, B), R_3(B)$. For every vertex $u \in U$, we include a tuple $t_u = (u)$ in $R_1(A)$; similarly, for every vertex $v \in V$, we include a tuple $t_v = (v)$ in $R_3(B)$. For every edge $(u, v) \in E$ where $u \in U, v \in V$, we include a tuple $t_{uv} = (u, v)$ in $R_2(A, B)$. Therefore the output tuples in $Q_{2-path}(D)$ corresponds to the edges in $E$. 
 \par
 We can see that PVCB has a solution of size $C$ if and only if $\ourprob(Q_{2-path}, k, D)$ has a solution of size $C$ for the same $k$. The only if direction is straightforward. For the other direction, note that by removing a tuple of the form $t_{uv}$ exactly one tuple from the output can be removed. Hence if any such tuple is chosen by the solution of \ourprob, it can be replaced by either $t_u$ or $t_v$ without increasing cost or decreasing the number of output tuples deleted. 
\end{proof}
}

We use a reduction from the PVCB problem when all the simplification steps fail in Algorithm~\ref{algo:dichotomy}. An example, proved in Appendix~\ref{sec:path-2}, is the query $Q_{2-path}(A, B) :- R_1(A), R_2(A, B), R_3(B)$ for paths of length two, where $A, B$ can correspond to $U, V$ for an easy reduction. However, the relations and attributes can form a more complex pattern like $R_1(A, B), R_2(B,H), R_3(C, E, B), R_4(E, H, A)$. The goal is to show that when the algorithm return false, there is a way to assign edges and vertices to the relations in $Q$. 

\begin{lemma}\label{lem:properties}
If none of the simplification steps in Algorithm~\ref{algo:dichotomy} can be applied for an intermediate query $Q$, all the following hold:
\begin{itemize}
\item[(1)] $Q$ has at least three relations, all with at least one attribute.
\item[(2)] There are no common attributes in all relations.
\item[(3)] For any two attributes $A, B \in \attr(Q)$, $\rel(A) \neq \rel(B)$, \ie, two attributes cannot belong to the exact same set of relations.
\item[(4)] The relations form a single connected components in terms of their attributes.
\end{itemize}
\end{lemma}
\begin{proof}
If these properties do not hold, simplification steps \casetwo\ (and \casesingle), \casecommon, \casecooccur, and \casedecompose\ can still be applied.
\end{proof}

Next we argue that if none of the simplification steps can be applied, there is a reduction from the partial vertex cover (PVCB) problem described in Lemma~\ref{lem:chain-3-nphard}. We divide the proof in two cases. First in Lemma~\ref{lem:nphard-disjoint} we show that if there are two relations in $Q$ that do not share any attribute, then $\ourprob$ is NP-hard for $Q$. Then in Lemma~\ref{lem:nphard-overlap} we show that if this condition does not hold, \ie, if for any two relations in $Q$ there is a common attribute, even then $\ourprob$ is NP-hard for $Q$.

\begin{lemma}\label{lem:nphard-disjoint}
If none of the simplification steps in Algorithm~\ref{algo:dichotomy} can be applied for an intermediate query $Q$, and if there are two relations $R_i, R_j$ in $\rel(Q)$ such that $\attr(R_i) \cap \attr(R_j) = \emptyset$, then $\ourprob(Q, k, D)$ is NP-hard for some $k, D$. 
\end{lemma}

\begin{proof}
Consider two relations $R_i, R_j$ in $\rel(Q)$ such that $\attr(R_i) \cap \attr(R_j) = \emptyset$. We give a reduction from the PVCB problem (Definition~\ref{def:pvc}). Let $G(U, V, E)$ and $k$ be the input to the PVCB problem. The value of $k$ remains the same and we create instance $D$ as follows. For every vertex $u \in U$, we include a tuple $t_u = (u, u, \cdots, u)$ to $R_i$ such that for every attribute $A \in \rel(R_i)$, $t_u.A = u$. Similarly, for every vertex $v \in V$, we include a tuple $t_v = (v, v, \cdots, v)$ to $R_j$ such that for every attribute $B \in \rel(R_j)$, $t_v.B = v$. 
\par
Now we grow two connected components of relations $C_i$ and $C_j$ starting with $C_i = \{R_i\}$ and $C_j = \{R_j \}$. We do the following until no more relations can be added to either $C_i$ or $C_j$. Consider any relation $R_{\ell}$ that has not been added to $C_i$ or $C_j$. (a) If $R_\ell$ shares an attribute $S$ with some relation in $C_i$ \emph{and} an attribute $T$ with some relation in $C_j$, it gets tuples $t_{uv}$ for every edge $(u, v) \in E$, where $t_{uv}.S = u$ and $t_{uv}.T = v$; for any other attribute $C \notin \attr(C_i) \cup \attr(C_j)$, it gets $t_{uv}.C = u$. (b) Otherwise, if $R_\ell$ shares an attribute only with some relation in $C_j$ \emph{but} no attributes with any relation in $C_i$, it gets tuples $t_{v}$ for every vertex $v \in V$ where $t_{v}.S = v$ for all attributes $S$ in $R_{\ell}$; further we add $R_{\ell}$ to $C_j$. (c) Otherwise, if $R_{\ell}$ shares an attribute only with some relation in $C_i$ (but not with any relation in $C_j$), then $R_{\ell}$ gets tuples $t_{u}$ for every vertex $u \in U$ where $t_{u}.S = u$ for all attributes $S$ in $R_{\ell}$; further, we add $R_{\ell}$ to $C_i$. 
\par
An example construction is shown in Figure~\ref{fig:reduc_from_PVCB}.
\par
Note that every relation in $Q$ corresponds to either vertices in $U$ or $V$, or edges in $E$. Clearly at least $R_i$ corresponds to $U$ and $R_j$ corresponds to $V$. We argue that at least one table $R_{\ell}$ has tuples $t_{uv}$ corresponding to the vertices: this holds from Lemma~\ref{lem:properties} since there are at least three relations, and the relations form a single connected components, otherwise relations that belong to $C_i$ and the ones that belong to $C_j$ would form two separate connected components.

\begin{figure}[!ht]
    \centering
        \begin{tikzpicture}[
            baseline=-20pt,
            node distance = 4mm and 12mm,
            start chain = going below,
            V/.style = {circle, draw, 
                fill=#1, 
                inner sep=0pt, minimum size=1mm,
                node contents={}},
                every fit/.style = {ellipse, draw=#1, inner ysep=1mm, 
                inner xsep=5mm}
            ]
            \foreach \i in {1, ..., 5} 
            {
                \ifnum\i=5
                    \node (n2\i) [V=mygreen, below right=2mm and 13mm of n14,
                                  label={[text=mygreen]right:$v_{\i}$}];
                \else
                    \node (n1\i) [V=myblue,on chain,
                                  label={[text=myblue]left:$u_{\i}$}];
                    \node (n2\i) [V=mygreen, above right=3mm and 13mm of n1\i,
                                  label={[text=mygreen]right:$v_{\i}$}];
                \fi
            }
            \node [myblue,fit=(n11) (n14),label=above:$U$] {};
            \node [mygreen,fit=(n21) (n25),label=above:$V$] {};
            \draw[-, shorten >=1mm, shorten <=1mm]{
                    (n11) edge (n21)    (n11) edge (n22)
                    (n12) edge (n23)
                    (n13) edge (n24)    (n13) edge (n25)
                    (n14) edge (n24)    (n14) edge (n25)};
        \end{tikzpicture}
        \\
        \begin{tabular}{| c |}
            \multicolumn{1}{c}{$R_1$}\\
    		\hline A \\\hline\hline
    		$u_1$\\
    		$u_2$\\
    		$u_3$\\
    		$u_4$\\
    		\hline
    	\end{tabular}
    	\begin{tabular}{|c|}
    	    \multicolumn{1}{c}{$R_2$}\\
    	    \hline B \\\hline\hline
    	    $v_1$ \\
    	    $v_2$ \\
    	    $v_3$ \\
    	    $v_4$ \\
    	    $v_5$ \\
    	    \hline
    	\end{tabular}
    	\begin{tabular}{|c|c|}
    	    \multicolumn{2}{c}{$R_3$}\\
    	    \hline A & C\\\hline\hline
    	    $u_1$ & $u_1$\\
    	    $u_2$ & $u_2$\\
    	    $u_3$ & $u_3$\\
    	    $u_4$ & $u_4$\\
    	    \hline
    	\end{tabular}
    	\begin{tabular}{|c|c|}
    	    \multicolumn{2}{c}{$R_4$}\\
    	    \hline E & B\\\hline\hline
    	    $v_1$ & $v_1$\\
    	    $v_2$ & $v_2$\\
    	    $v_3$ & $v_3$\\
    	    $v_4$ & $v_4$\\
    	    $v_5$ & $v_5$\\
    	    \hline
    	\end{tabular}
    	\begin{tabular}{|c|c|}
    	    \multicolumn{2}{c}{$R_5$}\\
    	    \hline C & E\\\hline\hline
    	    $u_1$ & $v_1$\\
    	    $u_1$ & $v_2$\\
    	    $u_2$ & $v_3$\\
    	    $u_3$ & $v_4$\\
    	    $u_3$ & $v_5$\\
    	    $u_4$ & $v_4$\\
    	    $u_4$ & $v_5$\\
    	    \hline
    	\end{tabular}
    	\begin{tabular}{|c|c|}
    	    \multicolumn{2}{c}{$R_6$}\\
            \hline C & F\\\hline\hline
            $u_1$ & $u_1$\\
    	    $u_2$ & $u_2$\\
    	    $u_3$ & $u_3$\\
    	    $u_4$ & $u_4$\\
    	    \hline
    	\end{tabular}
    \caption{An example construction from Lemma~\ref{lem:nphard-disjoint}: for the given instance of PVCB in the figure, and query $Q_2(\cdots) :-$ $R_1(A)$, $R_2(B)$, $R_3(A, C)$, $R_4(E, B)$, $R_5(C, E),$ $R_6(C, F)$ from the introduction, we create $D$ for $\ourprob(Q_2, k, D)$ as shown. First $R_1 \leftarrow U, R_2 \leftarrow V.$ Then $R_3 \leftarrow U, R_4 \leftarrow V$. Then $R_5 \leftarrow E$ and $R_6 \leftarrow U$.}
    \label{fig:reduc_from_PVCB}
\end{figure}
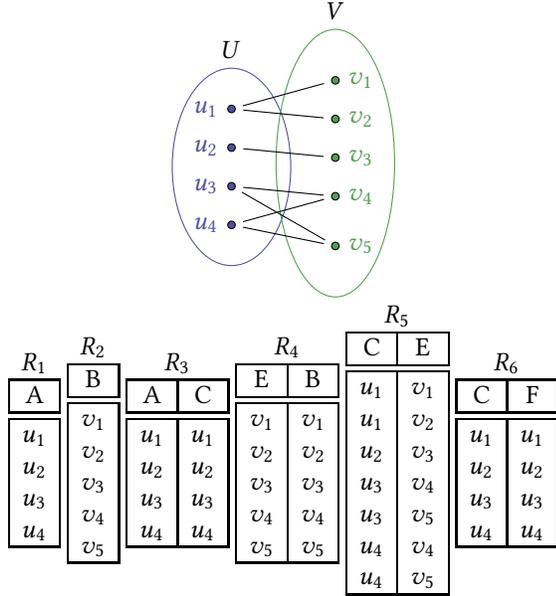

\par
 Now we claim that $PVCB$ has a solution of size $M$ if and only if $\ourprob(Q, k, D)$ has a solution of size $M$. Note that the output tuples in $Q(D)$ correspond to the edges in $E$.
\par 
(only if) If $PVCB$ has a solution of size $M$, we can remove the corresponding tuples from $R_i$ and $R_j$, and remove at least $k$ output tuples corresponding to the edges.
\par
(if) If $\ourprob$ has a solution of size $M$, we can assume wlog. that the input tuples are only chosen from $R_i$ and $R_j$:  if an input tuple $t_u$ is chosen from $R_{\ell} \in C_i$ (respectively, $C_j$), we replace it with the corresponding tuple in $R_i$ (respectively, $R_j$). If an input tuple $t_{uv}$ is chosen from a relation that shares attributes with both $C_i$ and $C_j$, we replace it with the corresponding $t_u$ from $R_i$. This removes at least the original tuples as before without increasing the cost. Now tuples from $R_i$ and $R_j$ corresponds to a solution of the PVCB problem that removes at least $k$ edges corresponding to the output tuples removed in $\ourprob(Q, k, D)$. 
\end{proof}

Next, we show the NP-hardness for the other case when any two relations share at least one attribute. We again give a reduction from the PVCB problem, where the input is a bipartite graph $G(U, V, E)$ and integer $k$. Given the relations in $Q$, we identify three relations where the tuples can correspond to $U, V, $ and $UV$ respectively. However, unlike the reduction in Lemma~\ref{lem:nphard-disjoint}, we may assign a constant value $*$ to all the tuples for some attributes in some tables. Due to the problem stated before Lemma~\ref{lem:hard-decompose}, we will ensure that no table in $Q$ receives such a constant value for all attributes. Otherwise this table will have a single tuple $(*, *, \cdots, *)$, and removing this tuple, will remove all output tuples from $Q(D)$ with only cost 1. We aim to identify the following (the relation names are chosen wlog.): (i) a relation $R_1$ corresponding to $U$, where attributes correspond to $U$ or $*$ (and not $V$), and at least one attribute corresponds to $U$, (ii)  a relation $R_3$ corresponding to $V$, where attributes correspond to $V$ or $*$ (and not $U$), and at least one attribute corresponds to $V$, and (iii) a relation $R_2$ corresponding to $E$, where attributes correspond to $U, V$ or $*$, and there are at least two attributes corresponding to $U$ and $V$ that together capture the edges in $E$. The other relations can have attributes corresponding to $U, V, $ or $*$, but no relation can have only attributes that take the constant value $*$. We give the formal reduction below by proving the following lemma.

\begin{lemma}\label{lem:nphard-overlap}
If none of the simplification steps in Algorithm~\ref{algo:dichotomy} can be applied for an intermediate query $Q$, and if for any two relations $R_i, R_j$ in $\rel(Q)$ it holds that $\attr(R_i) \cap \attr(R_j) \neq \emptyset$, then $\ourprob(Q, k, D)$ is NP-hard for some $k, D$. 
\end{lemma}

\begin{proof}
 We give a reduction from the PVCB problem, where the input is a bipartite graph $G(U, V, E)$ and integer $k$ (see Definition~\ref{def:pvc}). 
 \par
 We first observe the following property in addition to the properties $(1)-(4)$ in Lemma~\ref{lem:properties}.
 \begin{itemize}
 \item[(P1)] \emph{Any relation in $Q$ has at least two attributes.}
 \end{itemize}
 Suppose not, i.e., there is only one attribute $A$ in $R_i$. Since $R_i$ shares attributes with all relations in $Q$, $A$ appears in all relations in $Q$, violating property (2) from Lemma~\ref{lem:properties}.
 
 Recall that $\allattr_i = \attr(R_i)$ denotes the attributes in $R_i$. We also use 
 $$\allattr_{ij} = \allattr_i \cap \allattr_{j}$$
 to denote the the common attributes in $R_i$ and $R_j$, where $i < j$. Note that by property (1) of Lemma~\ref{lem:properties}, there are at least three relations in $Q$.
 
 \begin{itemize}
 \item[(P2) ] \emph{There exist three relations, wlog., $R_1, R_2, R_3$, and two attributes $A, B$ in $Q$  such that $A \in \allattr_{12} \setminus \allattr_{23}$ and $B \in \allattr_{23} \setminus \allattr_{12}$.} In other words, $A$ belongs to $R_1, R_2$ but not in $R_3$, and $B$ belongs to $R_2, R_3$ but not in $R_1$.  
 \end{itemize}
 
 To see (P2), start with the relation with the \textbf{smallest number of attributes} as $R_1$, breaking ties arbitrarily. Consider its intersection with all other relations (all are non-empty by assumption), and let $R_2$ be the relation with smallest number of attributes in the intersection $\allattr_{12}$ with $R_1$. If  any attribute in $\allattr_{12}$ belongs to $\allattr_{2j}$ for all $j > 2$, then it violates property (2) in Lemma~\ref{lem:properties}. Therefore, for all $A \in \allattr_{12}$, there is a relation $R_j$ such that $A \notin \allattr_{2j}$. Pick any such $A$ and the corresponding $R_j$. Now consider $\allattr_{2j}$. We claim that there exists $B \in \allattr_{2j} \setminus \allattr_{12}$. Suppose not. Then $\allattr_{2j} \subseteq \allattr_{12}$. Combining with the fact that there is an $A \in \allattr_{12} \setminus \allattr_{2j}$, $\allattr_{2j} \subset \allattr_{12}$  (a proper subset). This violates the assumption that $R_2$ is the relation with smallest number of attributes in the intersection of $\allattr_{12}$ with $R_1$. For simplicity, we assume $R_j = R_3$ wlog.
 \par
 Next we give the reduction from the PVCB problem by creating an instance $D$ for the same $k$ of \ourprob. $R_1, R_3$ correspond to $U, V$ respectively, whereas $R_2$ corresponds to the edges in $E$. 
 \begin{itemize}
\item  \emph{We include a tuple $t_u$ for each $u \in U$ to $R_1$ in $D$, where (i) for all $C \in \allattr_{13}$, $t_u.C = *$ (all attributes in $\allattr_{13}$ are constant attributes), and (ii) for all $C \in \allattr_1 \setminus \allattr_{13}$, $t_u.C = u$. 
 }
 \item \emph{We include a tuple $t_v$ for each $v \in V$ to $R_3$ in $D$, where (i) for all $C \in \allattr_{13}$, $t_v.C = *$, and (ii) for all $C \in \allattr_3 \setminus \allattr_{13}$, $t_v.C = v$.
 }
\end{itemize}
  Note that the assignment of values $U, V$ to attributes above is consistent, \ie, no attribute can get both $U$ and $V$. The above assignment is propagated to all other relations $R_j, j \neq 1, 3$, including $R_2$ (and by assigning any non-assigned attributes to $U$) as follows. For any other $R_j$, where $j \neq 1, 3$,
  \begin{itemize}
\item  If $\allattr_j$ includes an attribute $C \in \allattr_{13}$, it gets constant values in all tuples.
\item If $\allattr_j$ includes an attribute $C \in \allattr_1 \setminus \allattr_{13}$, it gets values corresponding to $U$.
\item If $\allattr_j$ includes an attribute $C \in \allattr_3 \setminus \allattr_{13}$, it gets values corresponding to $V$.
\item If $\allattr_j$ includes an attribute $C \notin \allattr_1 \cup \allattr_{3}$, it gets values corresponding to $U$.
\end{itemize}
After this assignment, if attributes in $\allattr_j$ are assigned to only constant and $U$, we insert tuples of the form $t_u$ as in $R_1$.  If attributes in $\allattr_j$ are assigned to only constant and $V$, we insert tuples of the form $t_v$ as in $R_3$. If attributes in $\allattr_j$ are assigned to both $U$ and $V$ (and possibly some constant attributes), we insert tuples of the form $t_{u, v}$ for each edge $(u, v) \in E$ in the same way. Hence at least $R_2$ gets tuples of the form $t_{uv}$ corresponding to the edges. The output $Q(D)$ of the query corresponds to the edges in $E$, but there can be multiple attributes for the vertices $u, v$ of an edge $(u, v)$.

\begin{figure}[!ht]
    \centering
    \begin{tabular}{|c|c|c|c|c|}
        \multicolumn{5}{c}{$R_1$}\\
        \hline A & P1 & P2 & E & F\\\hline\hline
        $v_1$ & $v_1$ & $v_1$ & $*$ & $*$\\
        $v_2$ & $v_2$ & $v_2$ & $*$ & $*$\\
        $v_3$ & $v_3$ & $v_3$ & $*$ & $*$\\
        $v_4$ & $v_4$ & $v_4$ & $*$ & $*$\\
        $v_5$ & $v_5$ & $v_5$ & $*$ & $*$\\
       
        \hline
    \end{tabular}
    \begin{tabular}{|c|c|c|c|c|}
        \multicolumn{5}{c}{$R_2$}\\
        \hline B & P1 & P2 & E & F\\\hline\hline
        $u_1$ & $v_1$ & $u_1$ & $*$ & $*$\\
        $u_1$ & $v_2$ & $u_1$ & $*$ & $*$\\
        $u_2$ & $v_3$ & $u_2$ & $*$ & $*$\\
        $u_3$ & $v_4$ & $u_3$ & $*$ & $*$\\
        $u_3$ & $v_5$ & $u_3$ & $*$ & $*$\\
        $u_4$ & $v_4$ & $u_4$ & $*$ & $*$\\
        $u_4$ & $v_5$ & $u_4$ & $*$ & $*$\\
        \hline
    \end{tabular}
    \begin{tabular}{|c|c|c|}
        \multicolumn{3}{c}{$R_3$}\\
        \hline P1 & C1 & C2\\\hline\hline
        $v_1$ & $u_1$ & $u_1$\\
        $v_2$ & $u_1$ & $u_1$\\
        $v_3$ & $u_2$ & $u_2$\\
        $v_4$ & $u_3$ & $u_3$\\
        $v_5$ & $u_3$ & $u_3$\\
        $v_4$ & $u_4$ & $u_4$\\
        $v_5$ & $u_4$ & $u_4$\\
        \hline
    \end{tabular}~~~~
    \begin{tabular}{|c|c|c|c|}
        \multicolumn{4}{c}{$R_4$}\\
        \hline P2 & C1 & C3 & F\\\hline\hline
        $v_1$ & $u_1$ & $u_1$ & $*$\\
        $v_2$ & $u_1$ & $u_1$ & $*$\\
        $v_3$ & $u_2$ & $u_2$ & $*$\\
        $v_4$ & $u_3$ & $u_3$ & $*$\\
        $v_5$ & $u_3$ & $u_3$ & $*$\\
        $v_4$ & $u_4$ & $u_4$ & $*$\\
        $v_5$ & $u_4$ & $u_4$ & $*$\\
        \hline
    \end{tabular}~~~~
    \begin{tabular}{|c|c|c|}
        \multicolumn{3}{c}{$R_5$}\\
        \hline E & F & C1\\\hline\hline
        $*$ & $*$ & $u_1$\\
        $*$ & $*$ & $u_2$\\
        $*$ & $*$ & $u_3$\\
        $*$ & $*$ & $u_4$\\
        \hline
    \end{tabular}
    
    \caption{An example construction for Lemma~\ref{lem:nphard-overlap}. Consider the PVCB instance from Figure \ref{fig:reduc_from_PVCB} and $Q_2(\cdots):-$ $R_1(A, P1, P2, E, F)$, $R_2(B, P1, P2, E, F)$, $R_3(P1, C1, C2)$, $R_4(P2, C1, C3, F)$, $R_5(E, F, C1)$ from the instruction. We create an instance $D$  for $\ourprob(Q_2, k, D)$ as shown. Here $R_5$ gets picked as the relation with the smallest set of attributes as $\widehat{\mathbf{R_1}}$ in the reduction (with a hat and boldfaced). Now, $\mathbb{A}_3$ has the smallest intersection with $\mathbb{A}_5$, so $R_3$ is chosen as $\widehat{\mathbf{R_2}}$. The attribute $C$ in the intersection is chosen as $\widehat{\mathbf{A}}$. $C \notin \attr(R_1)$, hence $R_1$ is chosen as $\widehat{\mathbf{R_3}}$. The common attribute $P1$ of $R_1, R_3$ (not in $R_5$)is chosen as $\widehat{\mathbf{B}}$. The common attributes $E, F$ of $R_1, R_5$ are assigned $*$. $ C1$ gets $U$, and $A, P1, P2$ get $V$, every other attribute gets $U$.}
    \label{fig:nphard-overlap}
\end{figure}
\par
An example construction is shown in Figure~\ref{fig:nphard-overlap}.
\par
Now we argue that PVCB has a solution of size $P$ if and only if \ourprob\ has a solution of size $P$. 
\par
(only if) If PVCB has a solution $S$ of size $P$ that covers $\geq k$ edges in $G$, we remove the corresponding tuples $t_u, t_v$ from $R_1$ and $R_3$, which will remove the set of $\geq k$ tuples corresponding to these edges in the output.
\par
(if) If \ourprob\ has a solution of size $P$ that removes at least $k$ output tuples, we argue that wlog., we can assume that the tuples are removed only from $R_1$ and $R_3$. This holds because of the following property:
 \begin{itemize}
 \item[(P3) ] \emph{No relation $R_j$ in $Q$ can have all attributes assigned to constant value $*$.}
 \end{itemize}
 Otherwise, by construction, $\allattr_j \subseteq \allattr_{13} \subset \allattr_1$ (since at least $A \in \allattr_1 \setminus \allattr_{13}$). Therefore $R_j$ has fewer attributes than $R_1$ violating the assumption that $R_1$ is the relation with the smallest number of attributes. Hence, we have the following property:
  \begin{itemize}
 \item[(P4) ] \emph{All relations in $Q$ have tuples of the form $t_u$ for $U$, $t_v$ for $V$, or $t_{uv}$ for $E$.}
 \end{itemize}
 If any tuple is removed from a $R_j$ that has $t_u$-s, we replace it with the corresponding tuple from $R_1$, and if any tuple is removed from a $R_j$ that has $t_v$-s, we replace it with the corresponding tuple from $R_3$. If any tuple of the form $t_{uv}$ is removed, we replace it with $t_u$ from $R_1$ that removes at least the same number of output tuples as before without increasing the cost. Hence, in any solution of \ourprob, we can assume that tuples are removed only from $R_1, R_3$, which corresponds to a solution of the PVCB problem where we remove the corresponding vertices to cover the edges corresponding to the output tuples.
 \end{proof}

Combining the above results, we get the following lemma:

\begin{lemma}\label{lem:hardness}
If Algorithm~\ref{algo:dichotomy} returns \false\ for a query $Q_0$, \ie, $\isptime(Q_0) = \false$, there exists $k$ and $D$ such that $\ourprob(Q_0, k, D)$ is NP-hard. 
\end{lemma}
\begin{proof}
Consider the \rectree\ $T$ of $Q_0$. Note that the leaves of $T$ are \true\ or \false. Note that if $\isptime(Q) = \false$ for any intermediate query $Q$, there must exist a path from $Q$ to a \false\ leaf and vice versa. We apply induction on the length of the shortest path from an intermediate node to a leaf. For the base case, \ie, when the length = 1, the intermediate query $Q$ at the node has a \false\ child. If there are two relations in $\rel(Q)$ that do not share any attribute, by Lemma~\ref{lem:nphard-disjoint}, $\ourprob$ for $Q$ is NP-hard. Otherwise, \ie, if no such two relations exist, then by Lemma~\ref{lem:nphard-overlap}, $\ourprob$ for $Q$ is NP-hard.
\par
Now consider the path from an NP-hard query $Q$ with a \false\ child to the root $Q_0$. 
If at least one node along this path has $\geq 2$ children, \ie, if at least once \casedecompose\ has been invoked, $Q_0$ is NP-hard by Lemma~\ref{lem:hard-decompose-general}.
\par
Otherwise, the path from $Q_0$ to a \false\ leaf is unique, and the parent of the \false\ leaf is NP-hard by the base case. 
Suppose the induction hypothesis holds for the intermediate node $Q$, \ie, $Q$ is NP-hard, where the shortest path length to a \false\ leaf is $\ell$, and consider the intermediate node $Q'$ that is parent of $Q$, and for which the shortest path length to a is $\ell + 1$. If $Q$ is formed from $Q'$ by \casecommon, then by Lemma~\ref{lem:hard-common} $Q'$ is NP-hard. Otherwise,  $Q$ is formed from $Q'$ by \casecooccur,  and by Lemma~\ref{lem:hard-cooccur} $Q'$ is NP-hard, proving the hypothesis. Repeating this argument, $Q_0$ is again NP-hard.
\end{proof}

\subsection{Algorithms}\label{sec:algo}
If Algorithm~\ref{algo:dichotomy} returns true for a query $Q$, we can find an optimal solution of $\ourprob(Q, k, D)$ by running Algorithm~\ref{algo:computeopt}. For each of the simplification step except the trivial case of empty query, we give a procedure that optimally solves that case, possibly using subsequent calls to $\computeopt$. Note that the first trivial case can never be reached if the original query is non-empty, for instance, if $Q(A):-R_1(A), R_2(A), R_3(A)$, before \casecommon\ is applied, instead \casesubset\ will be invoked, directly returning true in Algorithm~\ref{algo:dichotomy} and returning an optimal solution in Algorithm~\ref{algo:computeopt}. In Section~\ref{sec:proc-details}, we discuss the helper procedures used in Algorithm~\ref{algo:computeopt}. The pseudocodes of these procedures may suggest that $\computeopt$ is invoked recursively within these procedures. However, to ensure polynomial data complexity, we solve the problem bottom-up, and instead of recursive calls, use look ups from these pre-computed values, which is discussed in Section~\ref{sec:poly-impl}. Due to space constraints, {\bf all the pseudocodes of Section~\ref{sec:proc-details} are given in the appendix.}
\begin{algorithm}[t]\caption{Computing the optimal solution of $\ourprob(Q, k, D)$}\label{algo:computeopt}
{\footnotesize
	\begin{codebox}
		\Procname{$\computeopt(Q, k, D)$}
		\li \If $\rel(Q) = \emptyset$ or $\attr(Q) = \emptyset$  \hfill \textbf{/* (\caseempty) */}
		\li /* \textit{this case is never reached for a non-empty query} */ 
		\li \Do \Return $\emptyset$ 
		\li \ElseIf $Q$ has one relation  \textbf{/* (\casesingle) */}
		\li \Do \Return $\singlerel(Q, k, D)$
		\li \ElseIf $Q$ has two relations  \textbf{/* (\casetwo) */}
		\li \Do \Return $\tworel(Q, k, D)$
		\li \ElseIf $\exists R_i \in \rel(Q)$ such that $\forall~ R_j \neq R_i \in \rel(Q)$, \\
		$\attr(R_i) \subseteq \attr(R_j)$  \textbf{/* (\casesubset) */}
		\li \Do \Return $\onesubset(Q, k, D, R_i)$
		\li \ElseIf $\exists A \in \attr(Q)$ such that \\ for all relations $R_i \in \rel(Q)$, $A \in \attr(R_i)$ \\  \textbf{/* (\casecommon) */}
		\li \Return $\commonattr(Q, k, D, A)$   
		\li \ElseIf  $\exists A, B \in \attr(Q)$ such that $\rel(A) = \rel(B)$ \\  \textbf{/* (\casecooccur) */}
		\li \Return $\cooccur(Q, k, D, A, B)$   
		\li \ElseIf $Q$ can be decomposed into \emph{maximal connected components}\\ (see text) $Q^1, \cdots, Q^s$ where $s \geq 2$ \\ \textbf{/* (\casedecompose) */}
		\li \Do \Return $\decompose(Q, k, D, Q^1, \cdots, Q^k)$ 
		\li \Else fail	
	\end{codebox}
	}
\end{algorithm}

\subsubsection{Details of the procedures in Algorithm~\ref{algo:computeopt}}\label{sec:proc-details}

\textbf{1. Procedure $\singlerel(Q, k, D)$.~} Suppose $Q(\allattr_i) :- R_i(\allattr_i)$ be the query. 
We return any $k$ tuples from relation $R_i$ as the solution. Since there is no joins, each output tuple corresponds to a unique input tuple, and we can remove any $k$ input tuples to remove $k$ output tuples. 
\par
\cut{
\begin{algorithm}[ht]\caption{when $Q$ has two relations}\label{algo:tworel}
{\footnotesize
	\begin{codebox}
		\Procname{$\tworel(Q, k, D)$}
		\li Let $Q(\allattr_1 \cup \allattr_2) :- R_1(\allattr_1), R_2(\allattr_2)$ (wlog.)
		\li \If $\allattr_1 \cap \allattr_2 = \emptyset$
		\li \Do Let $n_1, n_2$ be the number of tuples in $R_1, R_2$ in $D$
		\li \If $n_1 \leq n_2$ \Do
		\li \Return any $\ceil{\frac{k}{n_1}}$ tuples from $R_1$
		\li \Else
		\li \Return any $\ceil{\frac{k}{n_2}}$ tuples from $R_2$ \End
		\li \Else \Do
		\li Let $\allattr_{12} = \allattr_1 \cap \allattr_2$
		\li Let $\aval_1, \cdots, \aval_g$ be all the distinct value combinations \\of attributes in $\allattr_{12}$ in $D$
		\li Let $G_i = $ set of tuples $t$ in $R_1$ such that $t.\allattr_{12} = \aval_i$, \\and let $m_i = |G_i|$, for $i = 1$ to $g$. 
		\li Let $H_i = $ set of tuples $t$ in $R_2$ such that $t.\allattr_{12} = \aval_i$, \\ and let $r_i = |H_i|$, for $i = 1$ to $g$.
		\li For $i = 1$ to $g$, let $p_i = \max(m_i, r_i)$
		\li wlog. assume that $p_1 \geq p_2 \geq \cdots \geq p_g$ \\ (else sort and re-index)
		\li Set numtup = 0, $i = 1$, $O = \emptyset$
		\li \While ${\tt numtup} \leq k$ \Do
		\li  \If $m_i \leq r_i$ \Do
		\li	 $S_i = G_i$ 
		\li \Else $S_i = H_i$ \End
		\li Include any tuple $t$ from $S_i$ to ${\tt O}$. $S_i = S_i \setminus \{t\}$
		\li \If $S_i = \emptyset$ \Do
		\li $i = i+1$ \End
		\li ${\tt numtup} = {\tt numtup} + 1$ \End
		\End
		\li \Return ${\tt O}$
		\End
	\end{codebox}
	}
\end{algorithm}
}
\textbf{2. Procedure $\tworel(Q, k, D)$.~} Suppose $Q(\allattr_1 \cup \allattr_2) :- R_i(\allattr_1), R_j(\allattr_2)$ be the query (wlog.). The pseudocode is given in Algorithm~\ref{algo:tworel}. There can be two cases.
\par
(a) If $R_1, R_2$ do not share any attribute, \ie, $\allattr_1 \cap \allattr_2 = \emptyset$, all tuples from $R_1$ join with all tuples from $R_2$ to form $Q(D)$. Let $n_1, n_2$ be the number of tuples from $R_1, R_2$ respectively.  Suppose $n_1 \leq n_2$. Then any tuple in $R_1$ removes exactly $n_2$ tuples from the output, which is higher than the number of tuples that a tuple from $R_1$ removes. In particular, consider any optimal solution $OPT$ and suppose it includes $s_1$ tuples from $R_1$ and $s_2$ tuples from $R_2$ that together remove at least $k$ output tuples. Removing the overlaps, we have,
$N_{OPT} = s_1n_2 + s_2 n1- s_1s_2$. Consider another solution $S$ that replaces all $s_2$ of $R_2$-tuples from $OPT$ by $s_2$ tuples from $R_1$ that have not been chosen yet. Now the number of output tuples deleted $N_S = (s_1 + s_2)n_2 = N_{OPT} + (n_2 - n_1)s_2 + s_1s_2 \geq N_{OPT}$ since $n_2 \geq n_1$. Hence we always get an optimal solution by removing tuples from $R_1$, and any $\ceil{\frac{k}{n_1}}$ of $R_1$-tuples remove at least $k$ output tuples.
\par
(b) Otherwise, let $\allattr_{12} = \allattr_1 \cap \allattr_2$ be the common attributes in $R_1, R_2$. Let $\aval_1, \aval_2, \cdots, \aval_g$ be all the distinct value combinations of $\attrset_{12}$ in $D$. 
We partition $R_1$ and $R_2$ into $G_1, \cdots, G_g$ and $H_1, \cdots, H_g$ respectively based on these values of $\attrset_{12}$. Hence, when we fix any $\aval_i$, all $R_1$-tuples in $G_i$ join by a cross product with all $R_2$-tuples in $H_i$. Therefore, the number of output tuples in $Q(D)$ is  $\sum_{i = 1}^g m_i. n_i$, where $m_i = |G_i|, n_i = |H_i|$ for $i = 1$ to $g$. First we sort all these groups in decreasing order of `profits' $p_i = \max(m_i, n_i)$, which is the maximum number of output tuples removed by removing only tuple from each group. wlog., assume  $p_1 \geq p_2 \geq \cdots \geq p_g$. Consider any group $i$: if $m_i \leq n_i$, we call $R_1$ the \emph{better relation} of group $i$, else $R_2$ is better. We get profit $p_i$ by removing tuples from the better relation of group $i$. Following the argument in case (a), it is more beneficial to remove from the better relation. Furthermore, in any optimal solution $OPT$ if a tuple $t_j$ has been chosen from group $j > i$ skipping a tuple $t_i$ from the better relation of group $i$, $t_j$ can be replaced by $t_i$ without increasing the cost and removing no fewer than the original number of output tuples. Hence we can assume wlog., that the optimal solution greedily chooses from the better relation of the groups $1, 2, \cdots, g$ in this order, which is implemented in the algorithm.
\par

\cut{
\begin{algorithm}[ht]\caption{When the attributes $\allattr_i$ of $R_i$ form a subset of all other relations in $Q$}\label{algo:onesubset}
{\footnotesize
	\begin{codebox}
		\Procname{$\onesubset(Q, k, D, R_i)$}
		\li Let $\aval_1, \cdots, \aval_g$ be all the distinct value combinations \\of attributes in $\allattr_{i}$ in $R_i$ in $D$\\
		(and they correspond to $g$ tuples in $R_i$)
		\li For every $\aval_j$, compute the number $m_j$ of output tuples\\ $t$ in $Q(D)$ such that $t.\allattr_i = \aval_j$. 
		\li Wlog. assume that $m_1 > m_2 > \cdots > m_g$ for $\aval_1, \aval_2, \cdots, \aval_g$
		\li Let $s$ be the smallest index such that $\sum_{j = 1}^s m_j \geq k$
		\li \Return the tuples from $R_i$ that correspond to $\aval_1, \cdots, \aval_s$ 
	\end{codebox}
	}
\end{algorithm}
}

\textbf{3. Procedure $\onesubset(Q, k, D, R_i)$.} Here the attributes in  $R_i$ form a subset of all other relations in $Q$. The pseudocode is given in Algorithm~\ref{algo:onesubset}. For every tuple in $R_i$, we compute the number of output tuple it contributes to, and choose greedily from a decreasing order on these numbers until $k$ output tuples are chosen. 
\par
The algorithm returns optimal solution since given any optimal solution of $\ourprob$ in this case, we can assume wlog. that all the tuples in the optimal solution belong to $R_i$. If any tuple $t$ is chosen from $R_{\ell} \neq R_i$, we can choose the corresponding tuple $t' = t.\allattr_i$ from $R_i$ instead, without increasing the cost and decreasing the number of output tuples deleted. Therefore the procedure always chooses from this sorted list, in decreasing order on $m_j$.

\cut{
\begin{algorithm}[ht]\caption{When all relations in $Q$ have a common attribute $A$ (the actual poly-time implementation is discussed in Section \ref{sec:poly-impl})}\label{algo:commonattr}
{\footnotesize
	\begin{codebox}
		\Procname{$\commonattr(Q, k, D, A)$}
		\li Let $a_1, \cdots, a_g$ be all the values of $A$ in $D$.
		\li We partition $D$ into $D_1, \cdots, D_g$, where all tuples $t$ in all tables \\
		in $D_i$ have $t.A = a_i$, $i = 1$ to $g$.
		\li 	Create a table $\optcost[1 \cdots g][1 \cdots k]$ where $\optcost[i][\ell]$ \\ 
		denotes the optimal solution to $\ourprob(Q, \ell, D)$ where the input \\ tuples can only be 
		chosen from $D_1, \cdots, D_i$.
		The \\corresponding solutions are stored in $\optsol[1 \cdots g][1 \cdots k]$.
		\li \For $i = 1$ to $g$ \Do
			\li \For $s = 1$ to $k$ \Do
				\li $\optcost[i][s] = \optcost[i-1][s]$ (also set $\optsol$)
				\li \For $m = 1$ to $s-1$ \Do
						\li Let $S_{i, m} = \computeopt[Q, m, D_i]$
						\li Let $c_{i, m} = |S_{i, m}|$
						\li \If $\optcost[i][s] > \optcost[i-1][s - m] + c_{i, m}$ \Do
							\li $\optcost[i][s] = \optcost[i-1][s - m] + c_{i, m}$\\ (and update $\optsol$)
						\End		
				\End
			\End
		\End
		\li \Return $\optsol[g][k]$.
	\end{codebox}
	}
\end{algorithm}
}
\textbf{4. Procedure $\commonattr(Q, k, D, A)$.} Here the attribute $A$ belongs to all relations in $Q$. The pseudocode is given in Algorithm~\ref{algo:commonattr}. First we partition the instance $D$ into $D_1, \cdots, D_g$, corresponding to $a_1, \cdots, a_g$, which are the all possible values of $A$ in $D$. All tuples $t$ in all relations in $D_i$ have $t.A = a_i$. Note that $Q(D)$ is a disjoint union of $Q(D_1), \cdots, Q(D_i)$.
\par
Here we run a dynamic program to compute the optimal solution $\optsol$
and its cost $\optcost$. Here $\optcost[i][s]$ denotes the minimum number of input tuples to remove at least $s$ output tuples from $Q(D)$ where the input tuples can only be chosen from $D_1$ to $D_i$.
This problem shows an optimal sub-structure property and can be solved with the following dynamic program:
{\small
\begin{equation*}
    \begin{split}
        \optcost[i][s] = \min
        & \begin{cases}
            \optcost[i-1][s] \\
            \min_{m = 1}^{s-1}\Big \{ \optcost[i-1][s - m] + c_{i, m} \Big\} \\
        \end{cases}
    \end{split}
\end{equation*}
}
where $c_{i, m}$ denotes the minimum number of input tuples \emph{only from} $D_i$ that would remove at least $m$ output tuples from $Q(D_i)$.  In other words, the above rules say that, to remove at least $s$ output tuples from $Q(D)$ where the input tuples can only be chosen from $D_1, \cdots, D_i$, we can either choose input tuples from $D_1, \cdots, D_{i-1}$ that achieve this goal, or we can remove at least $m$ tuples $Q(D_i)$ by removing $c_{i, m}$ tuples only from $D_i$, and take its union with the optimal solution for the rest of the $s-m$ output tuples that have to be removed from $Q(D)$ by only removing tuples from $D_1, \cdots, D_{i-1}$. 
\par
Since $k$ is bounded by the number of output tuples in $Q(D)$ (polynomial in data complexity) and $g$ is bounded by the number of input tuples,  the number of cells in $\optcost$ and $\optsol$ is polynomial in data complexity. However, the procedure $\commonattr(Q, k, D, A)$ is invoked in combination with other procedures in Algorithm~\ref{algo:computeopt}, which makes the total complexity non-obvious. In Section~\ref{sec:poly-impl} we discuss how the entire Algorithm~\ref{algo:computeopt} can be implemented in polynomial data complexity. 

\par

\cut{
\begin{algorithm}[ht]\caption{When two attributes $A, B$ appear in the same set of relations in $Q$}\label{algo:cooccur}
{\footnotesize
	\begin{codebox}
		\Procname{$\cooccur(Q, k, D, A, B)$}
		\li Replace both $A, B$ by a new attribute $C \notin \attr(Q)$\\
		 in all relations where $A$ and $B$ appear
		\li Let the new query be $Q_{AB \rightarrow C}$
		\li  Initialize $D' = D$
		\li If $A, B \in \attr(R_i)$, replace all original tuple $t \in R_i$ in $D$ by \\ $t'$ in $D'$
		 such that $t'.C = (t.A, t.B)$, and \\ $t'.F = t.F$ for all other attributes $\neq A, B$ in $R_i$ 
		\li Let $S = \computeopt(Q_{AB \rightarrow C}, k, D')$ 
		\li If $S$ includes any tuple $t$ from any $R_i$ in $Q$ such that \\$A, B \in \rel(R_i)$, change 
		all such tuples to their original form \\by replacing $t.C = (a, b)$ to $t.A = a, t.B = b$
		\li \Return $S$
	\end{codebox}
	}
\end{algorithm}
}
\textbf{5. Procedure $\cooccur(Q, k, D, A, B)$.~}  Here $\rel(A) = \rel(B)$. The pseudocode is given in Algorithm~\ref{algo:cooccur}. We simply replace both $A, B$ with a new attribute $C$ in all these relations, and assign values $t.C = (a, b)$ where $t.A = a, t.B = b$ in all such relations. Then we call the $\computeopt$ procedure to get the optimal solution for this new query and instance. Since the inputs and outputs of the new query $Q_{AB \rightarrow C}$ and instance $D'$ have a one-to-one correspondence with the inputs and outputs of the original query $Q$ and instance $D$, an optimal solution (if it exists) of the latter gives an optimal solution to the former by changing the tuples back to their original form.

\cut{
\begin{algorithm}[ht]\caption{When $Q$ can be decomposed into $s > 1$ connected components of relations}\label{algo:decompose}
{\footnotesize
	\begin{codebox}
		\Procname{$\decompose(Q, k, D, Q^1, \cdots, Q^s)$}
		\li Set $Q_1 = Q^1$. 
		\li \For $i = 2$ to $s$ \Do
			\li Set $Q_2 = Q^i$.
			\li \textit{/*Compute $\optsol_{i, s}$ below for the optimal solution to}\\
			 \textit{remove at least $s$ tuples from the output of $Q_i = $ join of $Q^1, \cdots, Q^i$*/}
			\li Let $m_1 = |Q_1(D)|$, and $m_2 = |Q_2(D)|$
				\For $s = 1$ to $k$ \Do
					\li $\optsol_{i, s} = \min_{k_1, k_2 : k_1, k_2 \leq s~\textit{and}~ k_1m_2 + k_2m_1 - k_1k_2 \geq s}$\\
			\hfill $ |\optsol_{i-1, k_1}| + |\computeopt(Q_2, k_2, D_2)|$
				\End
			\li $Q_{i+1}$ = join of $Q_i$ and $Q^{i+1}$.
			\li $i = i+1$.
		\End
		\li \Return $\optsol_{s, k}$. 
	\end{codebox}
	}
\end{algorithm}
}
\textbf{6. Procedure $\decompose(Q, k, D, Q^1, \cdots$, $Q^s)$.} Here $Q^1, \cdots, Q^s$ form maximal connected components, \ie, no relation in $Q^i$ share any attribute with any relation in $Q^j$ for $i \neq j$. The pseudocode is given in Algorithm~\ref{algo:decompose}, which generalizes the first part of Algorithm~\ref{algo:tworel} when two relations combine by a cross-product. The main difference is that, when the sub-queries are arbitrary and not a single relation, then the optimal solution may be required to select tuples from both sides of a partition. Although the output tuples from both partitions still join by cross-product, we do not know apriori how many output tuples to remove from each partition. Consider two disjoint components $Q_1, Q_2$, and let $m_1, m_2$ be the size of $Q_1(D)$ and $Q_2(D)$ respectively. Note that if an optimal solution $OPT$ removes $k_1$ tuples from $Q_1(D)$ and $k_2$ tuples from $Q_2(D)$, the total number of tuples removed from the join of $Q_1, Q_2$ is $k_1 m_2 + k_2 m_1 - k_1 k_2$, taking into account the tuples removed in the overlap. Algorithm~\ref{algo:decompose} takes two components at a time and aims to find the optimal solution for an arbitrary $s \leq k$, therefore we go over all possible requirements $k_1, k_2$ from $Q_1(D)$ and $Q_2(D)$ that satisfy $k_1m_2 + k_2m_1 - k_1k_2 \geq s$. The overall polynomial data complexity of Algorithm~\ref{algo:computeopt} is discussed in Section~\ref{sec:poly-impl}. 

\cut{
\begin{example}\label{eg:decompose-both}
\red{give the example that we may have to delete from both sides}
\end{example}
}

\subsubsection{Poly-time implementation of Algorithm~\ref{algo:computeopt}}\label{sec:poly-impl}

Algorithm~\ref{algo:computeopt} recursively calls itself through the helper procedures, and calls $\computeopt(Q', k', D')$ for many intermediate queries and values of $k' \leq k$. Suppose Algorithm~\ref{algo:computeopt} is invoked for $\computeopt(Q_0, k, D)$.  To ensure polynomial running time in data complexity, we first build the \rectree\ $T$ of $Q_0$ using Algorithm~\ref{algo:dichotomy} as defined in Definition~\ref{def:rectree}. Since Algorithm~\ref{algo:dichotomy} runs on the schema of $Q$, it is trivially polynomial in data complexity. If there is a leaf with value \false\ we know that the query $Q$ is NP-hard, and Algorithm~\ref{algo:computeopt} would return fail. Otherwise, using this tree and instance $D$, we solve $\computeopt(Q', s, D)$ for each intermediate query $Q'$ and value $s = 1 \cdots k$ in a bottom-up pass. These solutions and their sizes are simply looked up in Algorithms~\ref{algo:commonattr}, \ref{algo:cooccur}, and \ref{algo:decompose} instead of running $\computeopt$ recursively. Since there are a constant number of nodes in the \rectree\ (in data complexity, depends only on the number of relations and attributes), the maximum value of $k$ is $|Q_0(D)|$ (polynomial data complexity), and given these values all the algorithms run in polynomial time. Hence, we get a polynomial running time of Algorithm\ref{algo:computeopt}.

\par
Together with Lemma~\ref{lem:hardness}, this proves Theorem~\ref{thm:dichotomy}.

%% file: approx.tex
\section{Approximations}
\label{sec:approx}

\newcommand{\primaldual}{{\textsc{PrimalDual}}}

In this section, we discuss approximations for optimal solutions to $\ourprob(Q, k, D)$, where full CQ $Q$ contains $p$ relations in its body, $D$ is a given instance of the schema and $k$ is the number of output tuples we want to intervene on. In particular, we give a $p$-approximation for a general setting of our problem.


Recall from Lemma \ref{lem:hardness} that all the simplification steps in Algorithm \ref{algo:dichotomy} fail when $\ourprob(Q, k, D)$ is NP-hard. In such cases, we can model the problem as an instance of the \emph{Partial Set Cover problem} ($k'$-PSC).

\begin{definition}
    Given a set of elements $U$, a collection of subsets $\mathcal{S}\subseteq 2^U$, a cost function on sets $c:\mathcal{S}\rightarrow \mathcal{Q}^+$ and a positive integer $k'$, the goal of the Partial Set Cover problem ($k'$-PSC) is to pick the minimum cost collection of sets from $\mathcal{S}$ that covers at least $k$ elements in $U$.
\end{definition}\label{def:kPSC}

Observe the similarity between $\ourprob(Q, k, D)$ and $k'$-PSC in that we want to pick the smallest number of input tuples that intervene on at least $k$ output tuples. If there is a cost associated with deleting a specific input tuple, the cost function $c$ can be used to reflect this. Thus, sets correspond to input tuples from relations in the body of CQ $Q$ and elements to output tuples in $Q(D)$. Also, $k'=k$.

In \cite{GKS04}, Gandhi et al. give a primal-dual algorithm for partial set cover that generalizes the classic primal dual algorithm for set cover. If every element appears in at most $p$ sets in $\mathcal{S}$, they obtain a $p$-approximation for the problem (see Theorem 2.1 in \cite{GKS04}). Via a simple approximation preserving reduction given below, we also obtain a $p$-approximation for $\ourprob(Q, k, D)$ (proof in Appendix~\ref{sec:thm-approx}).

\begin{theorem}\label{thm:f_approx}
    $\ourprob(Q, k, D)$ has a $p$-approximation, which can be computed via an approximation preserving reduction to $k'$-PSC in poly-time.
\end{theorem}

\cut{
\begin{proof}
    We prove that the reduction preserves the approximation guarantee in two steps: 1) given an instance of $\ourprob(Q, k, D)$, how to construct an instance of $k'$-PSC, and 2) given a solution to $k'$-PSC, how to recover a solution to $\ourprob(Q, k, D)$.
    
    Given the full CQ $Q$ containing $p$ relations in its body, namely $R_1, R_2, \cdots, R_p$, we create a set per input tuple in the $p$ relations, and an element per output tuple in $Q(D)$. Each set contains elements that correspond to the output tuples resulting from the join between the associated input tuple and tuples from other relations in $Q$. It is well-known that the natural join on $R_1, R_2, \cdots, R_p$ can be computed in poly-time. Moreover, exactly one tuple in each of the $p$ relations participates in the join operation that produces a particular output tuple. Therefore, each element in the $k'$-PSC instance belongs to exactly $p$ sets. As a result, the size of the $k'$-PSC instance that we create is polynomial in the data complexity of $\ourprob(Q, k, D)$. Moreover, there is a one-on-one correspondence between instances of the two problems.
    
    Lastly, given a $p$-approximate solution to $k'$-PSC, we recover a solution to $\ourprob(Q, k, D)$ by picking the tuples associated with the sets in the solution, say $I$. Observe that the sets in $I$ cover $k' = k$ elements in $U$. Thus, removing the corresponding input tuples from $\ourprob(Q, k, D)$ will intervene on at least $k$ output tuples.
\end{proof}

}

\newcommand{\eat}[1]{}

\eat{

\subsection{2-approximation for $Q(A, B, E1, E2, E3) :- R_1(A, E_1), R_2(A, B, E_2), R_3(B, E_3)$}

Next, we illustrate the viability of improved approximation guarantees beyond that of the above result by giving a 2-approximation algorithm for an instance of the problem defined on 3 tables, where tables 1 and 2 share join attributes collectively denoted $A$, and tables 2 and 3 share join attributes collectively denoted $B$. In addition to these two attributes, each table also has a set of attributes that are unique to it, which we call $E_1, E_2,$ and $E_3$ for the respective tables. While each of these labels may represent a set of different attributes, we will combine them into a single attribute without loss of generality. Thus, in the rest of the section, $A, B, E_1, E_2,$ and $E_3$ are individual attributes that define an instance of our problem. We note that while these instances are restrictive in that there are only 3 tables and we only perform a chain join operation, this is representative of a large fraction of joins performed in practice. Moreover, our problem continues to be NP-hard even on this restricted set of instances.

First, we construct a graphical view of the problem. To this end, define a bipartite (multi-)graph as follows. The two sides of the graph contain vertices that represent individual input tuples in tables $R_1$ and $R_3$. The number of edges between two vertices $(a, e_1) \in R_1$ and $(b, e_3)\in R_3$ is the number of different value of $e_2$ such that $(a, b, e_2) \in R_2$. In particular, if there is no tuple in $R_2$ containing the pair $(a, b)$, then there is no edge between $(a, e_1)$ and $(b, e_3)$ for any values of $e_1$ and $e_3$. Thus, every edge is associated with a unique tuple $(a, b, e_2)$ --- we say that the edge is {\em colored} by this tuple. Note that if we consider all the edges corresponding to a single color, i.e., a single triple $(a, b, e_2)$, then they form a complete bipartite graph between the following two sets of vertices: $\{(a, e_1) \in R_1\}$ and $\{(b, e_3) \in R_3\}$. 

Note that each edge corresponds to an output tuple, each vertex in $R_1$ and $R_2$ correspond to input tuples in the respective tables, and each edge color class corresponds to an input tuple in $R_2$. Thus, on this graph, our problem can be equivalently defined as follows: remove the smallest number of vertices and/or edge color classes so that the total number of edges removed is at least $k$. 

An interesting observation is that the graph is locally isometric for two vertices that only differ in their values of $e_1$ for the vertices in $R_1$, or in their values of $e_3$ for the vertices in $R_3$. In other words, two vertices which share the same value of $a$ in $R_1$, or of $b$ in $R_3$ have edges of exactly the same color to exactly the same set of vertices. Moreover, two color classes that share the same value of $(a, b)$ but differ in $e_2$ are also isometric in the sense they have edges between exactly the same set of vertices. Finally, observe that if we fix a color class, say $(a, b, e_2)$, then the edges form a complete bipartite graph between vertices $\{(a, e_1)\in R_1: e_1\in E_1\}$ and  $\{(b, e_3)\in R_3: e_3\in E_3\}$. 

We use these observations to establish a structural property on the optimal solution to our problem. We say that a solution is {\em integral} if:
\begin{itemize}
    \item it removes all or none of the vertices in $R_1$ that have the same value of $a$.
    \item it removes all or none of the vertices in $R_3$ that have the same value of $b$.
    \item it removes all of none of the edge color class in $R_2$ that have the same value of $(a, b)$.
\end{itemize}
Relaxing the requirement slightly, we say that a solution is {\em nearly integral} if the above conditions are violated for at most one vertex or one edge color class. 

\begin{lemma}
\label{lma:nearint}
There is an optimal solution to the problem that is nearly integral. 
\end{lemma}

Using the structural property in the previous lemma, we can further consolidate the structure of the bipartite graph. In $R_1$, we coalesce all vertices with the same value of $a$, i.e., the set of vertices $R_1(a) = \{(a, e_1)\in R_1: e_1\in E_1\}$ into a single vertex with weight equal to the cardinality of $R_1(a)$. Similarly, we coalesce all vertices with the same value of $b$ in $R_3$, i.e., the set of vertices $R_3(b) = \{(b, e_3)\in R_3: e_3\in E_3\}$ into a single vertex with weight equal to the cardinality of $R_3(b)$. Each edge that was previously incident on a vertex is now incident on the supervertex that this vertex has been coalesced into. Note that each edge color class, which was previously a complete bipartite graph, is now a set of parallel edges -- in particular, the edge color class $(a, b, e_1)$ is a  set of parallel edges between the vertices $R_1(a)$ and $R_3(b)$. We replace these parallel edges by a single edge and drop the edge color since every edge has a unique color at this point. Furthermore, the parallel edges between two vertices $R_1(a)$ and $R_3(b)$ represent the different values of $e_2 \in E_2$ for which $(a, b, e_2)$ is in $R_2$. We replace these parallel edges with a single edge representing $R_2(a, b) = \{(a, b, e_2)\in R_2: e_2\in E_2\}$ with weight equal to the cardinality of $R_2(b)$. To simplify the notation, let us call $w_a = |R_1(a)|$, $w_b = |R_3(b)|$, and $w_{a,b} = |R_2(a, b)|$.

On this new graph, consider the following linear program (LP):

\begin{center}
    Minimize $\sum_{a\in A} w_a x_a + \sum_{b\in B} w_b x_b + \sum_{a\in A, b\in B} w_{a, b} x_{a, b}$ \\
    such that 
    $y_{a, b} \leq x_{a, b} + x_a + x_b \quad \forall a\in A, b\in B$ \\
    $y_{a, b} \leq 1 \quad \forall a\in A, b\in B$ \\
    $\sum_{a\in A, b\in B} w_a w_{a, b} w_b y_{a, b} \geq k$
\end{center}

Lemma~\ref{lma:nearint} says that there is an optimal fractional solution 
to the above LP that is nearly integral. In this case, by nearly integral,
we mean it in the standard LP sense, i.e., at most one variable $x$ has a 
fractional value strictly between $0$ and $1$ while all the other variables 
$x$ have value either $0$ or $1$.

}

%% file: conclusions.tex
\section{Conclusions}
\label{sec:conclusions}
In this paper, we studied the generalized deletion propagation (\ourprob) problem for full CQs without self-joins, gave a dichotomy to decide whether \ourprob\ for a query is poly-time solvable for all $k$ and instance $D$, and also gave approximation results. Several open questions remain. First, it would be good to study the complexity for larger classes of queries involving projections and/or self-joins, other classes of aggregates like {\em sum}, and also instances with arbitrary weights on the input and output tuples. Another interesting direction is to understand the approximability of this problem even for the restricted class of full CQs without self-join. We showed that a $p$-approximation exists where $p$ is the number of tables in the query, but our study of this problem suggests that this bound is not tight. Whether a poly-time algorithm exists giving an absolute constant independent of the schema as the approximation factor remains an interesting open problem.

%% file: appendix.tex
\appendix

\allowdisplaybreaks

\section{Hardness propagation for simplification steps}\label{sec:hard-propagate}
In this section we show that for the last three simplification steps in Algorithm~\ref{algo:dichotomy} that calls the $\isptime$ function recursively, if the new query $Q'$ (or one of the new queries) is NP-hard, then the query $Q$ is NP-hard. While the proofs for the fifth and sixth steps are more intuitive, the proof for the seventh step needs a careful construction. 

\textbf{Hardness propagation \casecommon .~} The following lemma shows the hardness propagation for the fifth simplification step  in Algorithm~\ref{algo:dichotomy}.
\begin{lemma} \label{lem:hard-common}
Let $\exists A \in \attr(Q)$ such that for all relations $R_i \in \rel(Q)$, $A \in \attr(R_i)$, and 
let $Q' = Q_{-A}$ be the query formed by removing $A$ from each relation in $\rel(Q)$. If $\ourprob(Q', k, D')$ is NP-hard, 
then $\ourprob(Q, k, D)$ is NP-hard. 
\end{lemma}
\begin{proof}
Given an instance of $\ourprob(Q', k, D')$, we construct an instance of $\ourprob(Q, k, D)$ as follows. 
\par
First, we claim that no relation $R_i$ in $Q'$ can have empty set of attributes. Otherwise, $A$ was the single attribute in $R_i$, which also appeared in all other relations in $Q$, \ie, $\forall R_j,~ \attr(R_i) \subseteq \attr(R_j)$. Therefore, the condition for \casesubset\ is satisfied, which is checked before \casecommon, and would have returned \true. Therefore, all relations in $Q'$ has at least one attribute.
\par
For any relation $R_i' \in \rel(Q')$, if a tuple $t'$ appears in $D'^{R_i'}$, create a new tuple $t$ in $D^{R_i}$ such that $t.A = *$ (a fixed value for all tuples and all relations in attribute $A$), and for all other attributes $E$, $t.E = t'.E$ (there is at least one such $E$). 
Hence there is a one-to-one correspondence between the tuples in the output $Q(D)$ and $Q'(D')$, and also in the input $D$ and $D'$. Therefore, a solution to $\ourprob(Q, k, D)$ of size $C$ corresponds to a solution to $\ourprob(Q', k, D')$, and vice versa. \end{proof}


\textbf{Hardness propagation for \casecooccur.~} Next we show the hardness propagation for the sixth simplification step in Algorithm~\ref{algo:dichotomy}.
\begin{lemma}\label{lem:hard-cooccur}
Let $A, B \in \attr(Q)$ such that $\rel(A) = \rel(B)$. Let $Q' = Q_{AB \rightarrow C}$ be the query by replacing $A, B$ with a new attribute $C \notin \attr(Q)$ in all relations. If $\ourprob(Q', k, D')$ is NP-hard, 
then $\ourprob(Q, k, D)$ is NP-hard.
\end{lemma}
\begin{proof}
Given an instance of $\ourprob(Q', k, D')$, we construct an instance of $\ourprob(Q, k, D)$ as follows. Consider any relation $R_i' \in \rel(Q')$ such that $C \in \attr(R_i')$, the corresponding relation $R_i$ in $Q$ has both attributes $A, B$ instead of $C$. If a tuple $t'$ appears in $D'^{R_i'}$, create a new tuple $t$ in $D^{R_i}$ such that $t.A = t.B = t'.C$, and for all other attributes $E$, $t.E = t'.E$ (\ie, both $A, B$ attributes get the value of attribute $C$ in $D$). Hence there is a one-to-one correspondence between the tuples in the output $Q(D)$ and $Q'(D')$, and also in the input $D$ and $D'$. Therefore, a solution to $\ourprob(Q, k, D)$ of size $C$ corresponds to a solution to $\ourprob(Q', k, D')$, and vice versa. 
\end{proof}

\textbf{Hardness propagation for \casedecompose.~} Now we show the hardness propagation for the seventh simplification step. Unlike the above two steps, this step requires a careful construction. For instance, consider a query  $Q(A, B, E):- R_1(A), R_2(A, B), R_3(B), R_4(E)$, which can be decomposed into two connected components $Q^1(A, B) :- R_1(A), R_2(A, B), R_3(B)$ and $Q^2(E) :- R_4(E)$. As we show later in Lemma~\ref{lem:chain-3-nphard}, $\ourprob(Q^1, k', D')$ is NP-hard for some $k', D'$, so although $Q^2$ is easy (see Section~\ref{sec:algo}), $\ourprob(Q, k, D)$ should be NP-hard for some $k, D$. An obvious approach is to assign a dummy value for $E$ in $R_4(E)$ similar to Lemma~\ref{lem:hard-common} above. However, if the number of tuples in $R_4$ is one or small in $D$, $\ourprob(Q, k, D)$ gains advantage by removing tuples from $R_4$, thereby completely bypassing $Q^1$. Therefore, a possible solution is to use a large number of tuples in $R_4$ that do not give a high benefit to delete from $R_4$, \eg, if it has more than the number of tuples in the output of $Q^1$ on $D$ restricted to $R_1, R_2, R_3$. First, we show the hardness propagation for a single application of \casedecompose. 
 
\begin{lemma} \label{lem:hard-decompose}
Let $Q$ is decomposed into maximal connected components  $Q^1, \cdots, Q^s$ where $s \geq 2$. without loss of generality (wlog.), suppose $\ourprob(Q^1, k', D')$ is NP-hard. 
Then, $\ourprob(Q, k, D)$ is NP-hard.
\end{lemma}
\begin{proof}
Given $Q^1, k', D'$, we create $k, D$ as follows. In $D$, all relations in $Q^1$ retain the same tuples. For all relations $R_i \in \bigcup_{j = 2}^s \rel(Q^j)$, we create $L$ tuples as follows: let us fix $R_i$, and let $A_1, \cdots, A_u$ be the attributes in $R_i$. $R_i$ in $D$ contains $L$ tuples of the form $t_{\ell} = (a_{1, \ell}, a_{2, \ell}, \cdots, a_{u, \ell})$, for $\ell = 1$ to $L$. Similarly, we populate the other relations. Note that in all the relations $R_i$ that an attribute $A_h$ appears in, it has $L$ values $a_{h, 1}, \cdots, a_{h_L}$. Therefore, any connected component $Q^2, \cdots, Q^s$ except $Q^1$ has $L$ output tuples (the components are maximally connected) and each input tuple of a relation participates in exactly one output tuple \emph{within} the connected component. Since $Q^1, \cdots, Q^s$ are disjoint in terms of attributes, in the output of $Q$, the outputs of each connected component will join in cross products. Suppose $Q^1(D')$ has $P$ output tuples. Then the number of output tuples in $Q(D)$ is $P \cdot L^{s-1}$. We set $k = k' \cdot L^{s-1}$ and $L = P+1$. The size of $D$ is $|D'| + L(s-1)$. Since $P \leq |D'|^{p'}$ (where $p'$ is the number of relations in $Q^1$), the increase in size of the inputs in this reduction is still polynomial in data complexity. 
Now we argue that $\ourprob(Q^1, k', D')$ has a solution of size $C$ if and only if $\ourprob(Q, k, D)$ has a solution of size $C$.
\par
(only if) If by removing $C$ tuples from $\rel(Q^1)$ we remove $k'$ tuples from $Q^1(D')$, then by removing the same  $C$ tuples we will remove $k' \cdot L^{s-1}$ output tuples from $Q(D)$ by construction as the output tuples from the connected components join by cross product, and each connected component has $L$ output tuples. 
\par
(if) Consider a solution to $\ourprob(Q, k, D)$ that removes at least $k = k'. L^{s-1}$ tuples from $Q(D)$. Note that any tuple from any relation $R_i \in \rel(Q^j)$, $j \geq 2$, can remove exactly 1 output tuple from the output of connected component $Q^j$ that it belongs to. Therefore, it removes exactly $o_2 = P. L^{s-2}$ tuples from the output. On the other hand, since we do not have any projection, any tuple from any relation $R_i \in \rel(Q^1)$ removes at least one output tuple from $Q^1(D')$, therefore at least $o_1 = L^{s-1}$ output tuples from $Q(D)$. Since $L = P+1$, $o_2 < o_1$. Therefore, if the assumed solution to $\ourprob(Q, k, D)$ removes any input tuple from any relation belonging to $Q^2, \cdots, Q^s$, we can replace it by any input tuple from the relations in $Q^1$ that has not been removed yet without increasing the cost or decreasing the number of output tuples removed. Therefore, wlog. all removed tuples appear in relations in $Q^1$. Since $k = k'. L^{s-1}$ tuples are removed from $Q(D)$, each tuple in $Q^1(D')$ removes exactly $L^{s-1}$ tuples from $Q(D)$, and the set of tuples removed from $Q(D)$ by tuples from $Q^1(D')$ are disjoint, at least $k'$ tuples must be removed from $Q^1(D')$ which gives a solution of cost at most $C$. 
\end{proof}

Although the above proof requires an exponential blow-up in the size of the query and not the data, there may be multiple application of \casedecompose\ in Algorithm~\ref{algo:dichotomy} in combination with the other simplification steps. Therefore, we need to ensure that the size of the instance $D$ that we create from $D'$ is still polynomial in data complexity for the original query that we started with. 
\par
Using ideas from Lemmas~\ref{lem:hard-common}, \ref{lem:hard-cooccur}, and \ref{lem:hard-decompose}, below we argue if any application of \casedecompose\ yields a hard query in one of the components, then the 	query we started with (say $Q_0$) is hard. Such an argument was not needed for   \casecommon\ and \casecooccur\ since in Lemma~\ref{lem:hard-common} and \ref{lem:hard-cooccur} the reductions do not yield an increase in the size of database instance. 
\par

\begin{lemma}\label{lem:hard-decompose-general}
Let $Q_0$ be the query that is given as the initial input to Algorithm~\ref{algo:dichotomy}. For any intermediate query $Q'$ in the \rectree\ of $Q_0$, if $\ourprob(Q', k', D')$ is NP-hard, 
then $\ourprob(Q_0, k, D)$ is NP-hard. 
\end{lemma} 
\begin{proof}
Consider the \rectree\   $T$ in which the simplification steps have been applied from $Q_0$ to $Q'$. 
Now consider the node $Q'$ in $T$ such that $\ourprob(Q', k', D')$ is NP-hard. The instance $D'$ is defined on the relations and attributes in $Q'$. From $D'$, we need to construct an instance $D$ on the relations and attributes in $Q_0$.  
\par
Consider the path from $Q'$ to the root $Q_0$. The relations in $Q'$ can lose attributes from the corresponding relations in $Q_0$ only by steps \casecommon\ and  \casecooccur\ along this path. 
For the the attributes that were lost on the path from $Q_0$ to $Q'$, we populate the values bottom-up from $Q'$ to $Q_0$ as follows. The relations appearing in $Q'$ have the same number of tuples in $D$ and $D'$. Moreover, (i) if two variables $A, B$ are replaced by a variable $C$  by \casecooccur, both $A$ and $B$ get the same values of $C$ in the corresponding tuples, (ii) if a variable $A$ is removed by \casecommon, we replace it by a constant value $*$ in all tuples. Let $Q$ be the query formed by extending the relations in $Q'$ with attributes by this process at the root.  
\par
Note that the relations in any non-descendant and non-ancestor node $Q_{nad}$ will be disjoint from those in $Q'$, but they can share some attributes with $Q$ only  by \casecommon: \casecommon\ is applied before \casecooccur\ so multiple attributes co-occurring in all relations will be removed by \casecommon\ not by \casecooccur; further before \casecooccur\ can be applied, the decomposition step \casedecompose\ must be called at least once.  We take all the relations that do not appear in the ancestors and descendants of $Q'$, and do a maximal connected component decomposition on them excluding the attributes that are common with $Q$. Let $s$ be the number of connected components. 
The tables in the connected components each get $L$ tuples as in the construction of Lemma~\ref{lem:hard-decompose}: consider a relation $R_i  \notin \rel(Q)$. (a) if there is an attribute $A \in \attr(R_i) \cap \attr(Q)$, assign $A = *$ in all $L$ tuples, (b) for all other attributes say $(A_1, \cdots, A_u) \in \attr(R_i) \setminus \attr(Q)$,  $R_i$ in $D$ contains $L$ tuples of the form $t_{\ell} = (a_{1, \ell}, a_{2, \ell}, \cdots, a_{u, \ell})$, for $\ell = 1$ to $L$. Therefore, the number of output tuples in each connected component is $L$ and each input tuple from each connected component can remove exactly one output tuple from the component. 

\begin{figure}[!ht]
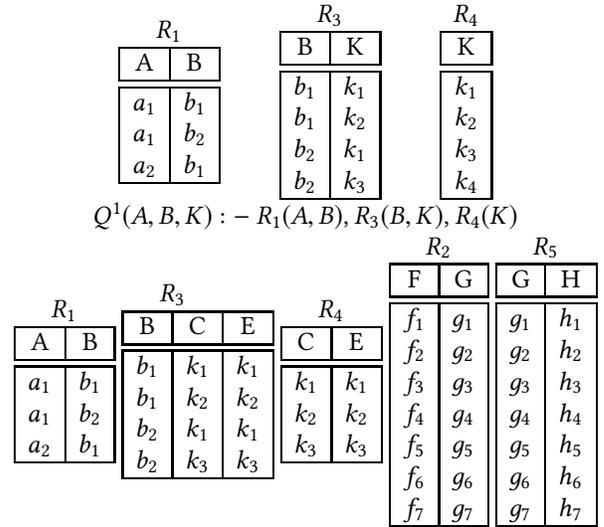

    \centering
    \begin{tabular}{|c|c|}
        \multicolumn{2}{c}{$R_1$}\\
        \hline A & B\\\hline\hline
        $a_1$ & $b_1$\\
        $a_1$ & $b_2$\\
        $a_2$ & $b_1$\\
        \hline
    \end{tabular}
    \qquad
    \begin{tabular}{|c|c|}
        \multicolumn{2}{c}{$R_3$}\\
        \hline B & K\\\hline\hline
        $b_1$ & $k_1$\\
        $b_1$ & $k_2$\\
        $b_2$ & $k_1$\\
        $b_2$ & $k_3$\\
        \hline
    \end{tabular}
    \qquad
    \begin{tabular}{|c|}
        \multicolumn{1}{c}{$R_4$}\\
        \hline K\\\hline\hline
        $k_1$\\
        $k_2$\\
        $k_3$\\
        $k_4$\\
        \hline
    \end{tabular}
    \\ $Q^1(A, B, K):-$ $R_1(A, B)$, $R_3(B, K)$, $R_4(K)$\\
    
    \begin{tabular}{|c|c|}
        \multicolumn{2}{c}{$R_1$}\\
        \hline A & B\\\hline\hline
        $a_1$ & $b_1$\\
        $a_1$ & $b_2$\\
        $a_2$ & $b_1$\\
        \hline
    \end{tabular}
    \begin{tabular}{|c|c|c|}
        \multicolumn{2}{c}{$R_3$}\\
        \hline B & C & E\\\hline\hline
        $b_1$ & $k_1$ & $k_1$\\
        $b_1$ & $k_2$ & $k_2$\\
        $b_2$ & $k_1$ & $k_1$\\
        $b_2$ & $k_3$ & $k_3$\\
        \hline
    \end{tabular}
    \begin{tabular}{|c|c|}
        \multicolumn{2}{c}{$R_4$}\\
        \hline C & E\\\hline\hline
        $k_1$ & $k_1$\\
        $k_2$ & $k_2$\\
        $k_3$ & $k_3$\\
        \hline
    \end{tabular}
    \begin{tabular}{|c|c|}
        \multicolumn{2}{c}{$R_2$}\\
        \hline F & G\\\hline\hline
        $f_1$ & $g_1$\\
        $f_2$ & $g_2$\\
        $f_3$ & $g_3$\\
        $f_4$ & $g_4$\\
        $f_5$ & $g_5$\\
        $f_6$ & $g_6$\\
        $f_7$ & $g_7$\\
        \hline
    \end{tabular}
    \begin{tabular}{|c|c|}
        \multicolumn{2}{c}{$R_5$}\\
        \hline G & H\\\hline\hline
        $g_1$ & $h_1$\\
        $g_2$ & $h_2$\\
        $g_3$ & $h_3$\\
        $g_4$ & $h_4$\\
        $g_5$ & $h_5$\\
        $g_6$ & $h_6$\\
        $g_7$ & $h_7$\\
        \hline
    \end{tabular}
    \caption{An example construction for Lemma~\ref{lem:hard-decompose-general}: Consider the full CQ $Q_0$ from Figure \ref{fig:example_T}. For the instance of $Q^1$ given in this figure, we create an instance $D$ for $Q_0$. 
    Since we replaced $C,E$ with $K$,
    both $C, E$ get the values from $K$. The output size of $Q^1$ is 6, therefore the relations in $Q^2$ that do not share any attributes with $Q^1$, so we add $6 + 1 = 7$ tuples to $R_2$ and $R_5$ of the form $(f_i, g_i)$ and $(g_i, h_i)$, for $1\leq i\leq 7$, respectively.}
    \label{fig:hard-decompose-general}
\end{figure}
An example reduction is shown in Figure~\ref{fig:hard-decompose-general}.
\par
Eventually, $Q_0$ is formed by joining $Q$ with the relations in the $s$ connected components. The attribute values $*$ or repeated values due to \casecommon\ and \casecooccur\ are not going to impact the number of output tuples.
\par
Now the same reduction as in Lemma~\ref{lem:hard-decompose} works: we set $k = k'. L^s$ where $L = P+1$, and $P = $ the number of tuples in $Q'(D')$. 
We again argue that $\ourprob(Q', k', D')$ has a solution of size $C$ if and only if $\ourprob(Q_0, k, D)$ has a solution of size $C$.
\par
(only if) If by removing $C$ tuples from $\rel(Q')$ we remove $k'$ tuples from $Q'(D')$, then by removing the corresponding $C$ tuples from $Q$, we will remove $k' \cdot L^s$ output tuples from $Q_0(D)$, by construction. The output tuples from the $s$ connected components join by cross product, and each connected component has $L$ output tuples. 
\par
(if) Consider a solution to $\ourprob(Q_0, k, D)$ that removes at least $k = k'. L^{s}$ tuples from $Q_0(D)$. Note that any tuple from any relation $R_i \notin \rel(Q)$, can remove exactly 1 output tuple from the output of connected component that it belongs to. Therefore, it removes exactly $o_2 = P. L^{s-1}$ tuples from the output. On the other hand, since we do not have any projection, any tuple from any relation $R_i \in \rel(Q)$ removes at least one output tuple from $Q'(D')$, therefore at least $o_1 = L^{s}$ output tuples from $Q_0(D)$. Since $L = P+1$, $o_2 < o_1$. Therefore, if the assumed solution to $\ourprob(Q_0, k, D)$ removes any input tuple from any relation belonging to the relations $\notin \rel(Q)$, we can replace it by any input tuple from the relations in $Q$ that has not been removed yet without increasing the cost or decreasing the number of output tuples removed. Therefore, wlog. all removed tuples appear in relations in $Q$. Since $k = k'. L^{s}$ tuples are removed from $Q_0(D)$, each tuple in the relations from $Q$ removes exactly $L^{s}$ tuples from $Q_0(D)$. Since the set of tuples removed from $Q_0(D)$ by tuples from $Q$ are disjoint, and the extension of attributes from relations in $Q'$ to those in $Q$ by repeating values or by using a constant $*$ does not have an effect on the number of output tuples, at least $k'$ tuples must be removed from $Q'(D')$, giving a solution of cost at most $C$. 
\end{proof}

\begin{algorithm}[ht]\caption{when $Q$ has two relations}\label{algo:tworel}
{\footnotesize
	\begin{codebox}
		\Procname{$\tworel(Q, k, D)$}
		\li Let $Q(\allattr_1 \cup \allattr_2) :- R_1(\allattr_1), R_2(\allattr_2)$ (wlog.)
		\li \If $\allattr_1 \cap \allattr_2 = \emptyset$
		\li \Do Let $n_1, n_2$ be the number of tuples in $R_1, R_2$ in $D$
		\li \If $n_1 \leq n_2$ \Do
		\li \Return any $\ceil{\frac{k}{n_1}}$ tuples from $R_1$
		\li \Else
		\li \Return any $\ceil{\frac{k}{n_2}}$ tuples from $R_2$ \End
		\li \Else \Do
		\li Let $\allattr_{12} = \allattr_1 \cap \allattr_2$
		\li Let $\aval_1, \cdots, \aval_g$ be all the distinct value combinations \\of attributes in $\allattr_{12}$ in $D$
		\li Let $G_i = $ set of tuples $t$ in $R_1$ such that $t.\allattr_{12} = \aval_i$, \\and let $m_i = |G_i|$, for $i = 1$ to $g$. 
		\li Let $H_i = $ set of tuples $t$ in $R_2$ such that $t.\allattr_{12} = \aval_i$, \\ and let $r_i = |H_i|$, for $i = 1$ to $g$.
		\li For $i = 1$ to $g$, let $p_i = \max(m_i, r_i)$
		\li wlog. assume that $p_1 \geq p_2 \geq \cdots \geq p_g$ \\ (else sort and re-index)
		\li Set numtup = 0, $i = 1$, $O = \emptyset$
		\li \While ${\tt numtup} \leq k$ \Do
		\li  \If $m_i \leq r_i$ \Do
		\li	 $S_i = G_i$ 
		\li \Else $S_i = H_i$ \End
		\li Include any tuple $t$ from $S_i$ to ${\tt O}$. $S_i = S_i \setminus \{t\}$
		\li \If $S_i = \emptyset$ \Do
		\li $i = i+1$ \End
		\li ${\tt numtup} = {\tt numtup} + 1$ \End
		\End
		\li \Return ${\tt O}$
		\End
	\end{codebox}
	}
\end{algorithm}

\section{\ourprob\ for Paths of Length-2 is NP-Hard}\label{sec:path-2}

\begin{lemma}\label{lem:chain-3-nphard}
For the query $Q_{2-path}(A, B) :- R_1(A), R_2(A, B), R_3(B)$, the problem $\ourprob(Q, k, D)$ is NP-hard.
\end{lemma}
\begin{proof}
We give a reduction from PVCB problem that takes as input $G=(U, V, E)$ and $k$.
 \par
 Given an instance of the PVCB problem, we construct an instance $D$ of \ourprob\ as follows for  $Q_{2-path}(A, B) :- R_1(A), R_2(A, B), R_3(B)$. For every vertex $u \in U$, we include a tuple $t_u = (u)$ in $R_1(A)$; similarly, for every vertex $v \in V$, we include a tuple $t_v = (v)$ in $R_3(B)$. For every edge $(u, v) \in E$ where $u \in U, v \in V$, we include a tuple $t_{uv} = (u, v)$ in $R_2(A, B)$. Therefore the output tuples in $Q_{2-path}(D)$ corresponds to the edges in $E$. 
 \par
 We can see that PVCB has a solution of size $C$ if and only if $\ourprob(Q_{2-path}, k, D)$ has a solution of size $C$ for the same $k$. The only if direction is straightforward. For the other direction, note that by removing a tuple of the form $t_{uv}$ exactly one tuple from the output can be removed. Hence if any such tuple is chosen by the solution of \ourprob, it can be replaced by either $t_u$ or $t_v$ without increasing cost or decreasing the number of output tuples deleted. 
\end{proof}

\begin{algorithm}[ht]\caption{When the attributes $\allattr_i$ of $R_i$ form a subset of all other relations in $Q$}\label{algo:onesubset}
{\footnotesize
	\begin{codebox}
		\Procname{$\onesubset(Q, k, D, R_i)$}
		\li Let $\aval_1, \cdots, \aval_g$ be all the distinct value combinations \\of attributes in $\allattr_{i}$ in $R_i$ in $D$\\
		(and they correspond to $g$ tuples in $R_i$)
		\li For every $\aval_j$, compute the number $m_j$ of output tuples\\ $t$ in $Q(D)$ such that $t.\allattr_i = \aval_j$. 
		\li Wlog. assume that $m_1 > m_2 > \cdots > m_g$ for $\aval_1, \aval_2, \cdots, \aval_g$
		\li Let $s$ be the smallest index such that $\sum_{j = 1}^s m_j \geq k$
		\li \Return the tuples from $R_i$ that correspond to $\aval_1, \cdots, \aval_s$ 
	\end{codebox}
	}
\end{algorithm}
\begin{algorithm}[ht]\caption{When all relations in $Q$ have a common attribute $A$ (the actual poly-time implementation is discussed in Section \ref{sec:poly-impl})}\label{algo:commonattr}
{\footnotesize
	\begin{codebox}
		\Procname{$\commonattr(Q, k, D, A)$}
		\li Let $a_1, \cdots, a_g$ be all the values of $A$ in $D$.
		\li We partition $D$ into $D_1, \cdots, D_g$, where all tuples $t$ in all tables \\
		in $D_i$ have $t.A = a_i$, $i = 1$ to $g$.
		\li 	Create a table $\optcost[1 \cdots g][1 \cdots k]$ where $\optcost[i][\ell]$ \\ 
		denotes the optimal solution to $\ourprob(Q, \ell, D)$ where the input \\ tuples can only be 
		chosen from $D_1, \cdots, D_i$.
		The \\corresponding solutions are stored in $\optsol[1 \cdots g][1 \cdots k]$.
		\li \For $i = 1$ to $g$ \Do
			\li \For $s = 1$ to $k$ \Do
				\li $\optcost[i][s] = \optcost[i-1][s]$ (also set $\optsol$)
				\li \For $m = 1$ to $s-1$ \Do
						\li Let $S_{i, m} = \computeopt[Q, m, D_i]$
						\li Let $c_{i, m} = |S_{i, m}|$
						\li \If $\optcost[i][s] > \optcost[i-1][s - m] + c_{i, m}$ \Do
							\li $\optcost[i][s] = \optcost[i-1][s - m] + c_{i, m}$\\ (and update $\optsol$)
						\End		
				\End
			\End
		\End
		\li \Return $\optsol[g][k]$.
	\end{codebox}
	}
\end{algorithm}
\begin{algorithm}[ht]\caption{When two attributes $A, B$ appear in the same set of relations in $Q$}\label{algo:cooccur}
{\footnotesize
	\begin{codebox}
		\Procname{$\cooccur(Q, k, D, A, B)$}
		\li Replace both $A, B$ by a new attribute $C \notin \attr(Q)$\\
		 in all relations where $A$ and $B$ appear
		\li Let the new query be $Q_{AB \rightarrow C}$
		\li  Initialize $D' = D$
		\li If $A, B \in \attr(R_i)$, replace all original tuple $t \in R_i$ in $D$ by \\ $t'$ in $D'$
		 such that $t'.C = (t.A, t.B)$, and \\ $t'.F = t.F$ for all other attributes $\neq A, B$ in $R_i$ 
		\li Let $S = \computeopt(Q_{AB \rightarrow C}, k, D')$ 
		\li If $S$ includes any tuple $t$ from any $R_i$ in $Q$ such that \\$A, B \in \rel(R_i)$, change 
		all such tuples to their original form \\by replacing $t.C = (a, b)$ to $t.A = a, t.B = b$
		\li \Return $S$
	\end{codebox}
	}
\end{algorithm}

\begin{algorithm}[ht]\caption{When $Q$ can be decomposed into $s > 1$ connected components of relations}\label{algo:decompose}
{\footnotesize
	\begin{codebox}
		\Procname{$\decompose(Q, k, D, Q^1, \cdots, Q^s)$}
		\li Set $Q_1 = Q^1$. 
		\li \For $i = 2$ to $s$ \Do
			\li Set $Q_2 = Q^i$.
			\li \textit{/*Compute $\optsol_{i, s}$ below for the optimal solution to}\\
			 \textit{remove at least $s$ tuples from the output of $Q_i = $ join of $Q^1, \cdots, Q^i$*/}
			\li Let $m_1 = |Q_1(D)|$, and $m_2 = |Q_2(D)|$
				\For $s = 1$ to $k$ \Do
					\li $\optsol_{i, s} = \min_{k_1, k_2 : k_1, k_2 \leq s~\textit{and}~ k_1m_2 + k_2m_1 - k_1k_2 \geq s}$\\
			\hfill $ |\optsol_{i-1, k_1}| + |\computeopt(Q_2, k_2, D_2)|$
				\End
			\li $Q_{i+1}$ = join of $Q_i$ and $Q^{i+1}$.
			\li $i = i+1$.
		\End
		\li \Return $\optsol_{s, k}$. 
	\end{codebox}
	}
\end{algorithm}

\section{Proof of Theorem~4.2}\label{sec:thm-approx}

\begin{proof}
    We prove that the reduction preserves the approximation guarantee in two steps: 1) given an instance of $\ourprob(Q, k, D)$, how to construct an instance of $k'$-PSC, and 2) given a solution to $k'$-PSC, how to recover a solution to $\ourprob(Q, k, D)$.
    
    Given the full CQ $Q$ containing $p$ relations in its body, namely $R_1, R_2, \cdots, R_p$, we create a set per input tuple in the $p$ relations, and an element per output tuple in $Q(D)$. Each set contains elements that correspond to the output tuples resulting from the join between the associated input tuple and tuples from other relations in $Q$. It is well-known that the natural join on $R_1, R_2, \cdots, R_p$ can be computed in poly-time. Moreover, exactly one tuple in each of the $p$ relations participates in the join operation that produces a particular output tuple. Therefore, each element in the $k'$-PSC instance belongs to exactly $p$ sets. As a result, the size of the $k'$-PSC instance that we create is polynomial in the data complexity of $\ourprob(Q, k, D)$. Moreover, there is a one-on-one correspondence between instances of the two problems.
    
    Lastly, given a $p$-approximate solution to $k'$-PSC, we recover a solution to $\ourprob(Q, k, D)$ by picking the tuples associated with the sets in the solution, say $I$. Observe that the sets in $I$ cover $k' = k$ elements in $U$. Thus, removing the corresponding input tuples from $\ourprob(Q, k, D)$ will intervene on at least $k$ output tuples.
\end{proof}